\documentclass[12pt]{article}
\usepackage{amsmath,amssymb,amsfonts,url,natbib}

\usepackage{amsthm}
\usepackage{enumerate}
\usepackage[utf8]{inputenc}
\usepackage{color}
\usepackage{graphicx}
\usepackage{hyperref}
\usepackage{caption,subcaption}
\usepackage[symbol]{footmisc}
\usepackage{algpseudocode}
\usepackage{algorithm}
\graphicspath{{Figures/}}

\newcommand{\RR}{\mathbb{R}}
\newcommand{\EE}{\mathbb{E}}
\newcommand{\PP}{\mathbb{P}}

\newtheorem{lemma}{Lemma}
\newtheorem{thm}{Theorem}
\newtheorem{corollary}{Corollary}

\theoremstyle{definition}
\newtheorem{example}{Example}
\newtheorem{definition}{Definition}
\newtheorem{remark}{Remark}

\newcommand{\R}{\mathbb{R}}

\newcommand{\bA}{\mathbf{A}}

\newcommand{\bP}{\mathbf{P}}

\newcommand{\bS}{\mathbf{S}}

\newcommand{\bU}{\mathbf{U}}

\newcommand{\bW}{\mathbf{W}}
\newcommand{\bX}{\mathbf{X}}

\newcommand{\bZ}{\mathbf{Z}}

\newcommand{\bx}{\mathbf{x}}
\newcommand{\by}{\mathbf{y}}

\newcommand{\supp}{\operatorname{supp}}

\newcommand{\blind}{0}

\begin{document}
\def\spacingset#1{\renewcommand{\baselinestretch}%
{#1}\small\normalsize} \spacingset{1}

%%%%%%%%%%%%%%%%%%%%%%%%%%%%%%%%%%%%%%%%%%%%%%%%%%%%%%%%%%%%%%%%%%%%%%%%%%%%%%

\if0\blind
{\title{\bf Euclidean mirrors and dynamics in network time series}
\author{Avanti Athreya\thanks{Co-first author. Department of Applied Mathematics and Statistics, Johns Hopkins University, Baltimore, MD 21218. Email: dathrey1@jhu.edu. All authors gratefully acknowledge funding from Microsoft Research, the Naval Engineering Education Consortium, the United States National Science Foundation (SES-1951005), and the Acheson Duncan Fund.}\hspace{.2cm}\\ and \\Zachary Lubberts\thanks{Co-first author, Department of Statistics, University of Virginia, Charlottesville, VA 22904. Email: zlubberts@virginia.edu.}\\ and \\Youngser Park\thanks{Center of Imaging Science, Johns Hopkins University, Baltimore, MD 21218. Email: youngser@jhu.edu.}\\ and \\ Carey Priebe\thanks{Department of Applied Mathematics and Statistics, Johns Hopkins University, Baltimore, MD 21218. Email: cep@jhu.edu.}}
\maketitle
} \fi

\if1\blind
{
  \bigskip
  \bigskip
  \bigskip
  \begin{center}
    {\LARGE\bf Euclidean mirrors and dynamics in network time series}
\end{center}
  \medskip
} \fi
\bigskip
\begin{abstract}
Analyzing changes in network evolution is central to statistical network inference, as underscored by recent challenges of predicting and distinguishing pandemic-induced transformations in organizational and communication networks. We consider a joint network model in which each node has an associated time-varying low-dimensional latent vector of feature data, and connection probabilities are functions of these vectors. Under mild assumptions, the time-varying evolution of the latent vectors exhibits low-dimensional manifold structure under a suitable notion of distance. This distance can be approximated by a measure of separation between the observed networks themselves, and there exist Euclidean representations for underlying network structure, as characterized by this distance, at any given time. These Euclidean representations, called Euclidean mirrors, permit the visualization of network evolution and transform network inference questions such as change-point and anomaly detection into a classical setting. We illustrate our methodology with real and synthetic data, and identify change points corresponding to massive shifts in pandemic policies in a communication network of a large organization.

% We show that our Euclidean representation of network evolution can be mapped to a flow or a curve, in which change points correspond to    for settings where the latent positions change slowly over time (perhaps except at a select number of change-points, where the latent positions undergo a rapid shift), the latent position random variables will lie on a manifold that can be embedded into a low-dimensional Euclidean space, where we can consistently estimate their representative points. This gives a visualization for the generating mechanism of the networks over time, revealing their temporal dynamics.

% modeling the temporal dynamics of networks by representing entities with time-varying latent position vectors, where the probabilities of communication between entities at a given time is a function of their associated latent positions. Under mild assumptions on the latent position trajectories, there is a one-dimensional manifold structure for the evolution of the constellation of latent positions in the organization over time, under a suitable notion of distance. We develop consistent estimation procedures for this manifold, enabling visualization of the communication network dynamics, and change-point and anomaly detection. We then apply these methods to real communication networks, recovering the pandemic change-point.
\end{abstract}
\noindent%
\vfill

\newpage
\spacingset{1.9} % DON'T change the spacing!

\section{Introduction}\label{sec:Intro}
The structure of many organizational and communication networks underwent a dramatic shift during the disruption of the COVID-19 pandemic in 2020 \citep{zuzul2021dynamic}. This massive shock altered network connectivity in many respects and across multiple scales, differentially impacting individual nodes, local sub-communities, and whole networks. A visualization of this can be seen in Figure~\ref{fig:org_sci_example}, which illustrates the shifting structure, from the spring to the summer of 2020, in a communications network of a large corporation. Each node represents an email account, and connection between nodes reflect email frequency between accounts. The panel on the left of Figure~\ref{fig:org_sci_example} shows a clustering of the network into subcommunities, and the panel on the right shows how those network connections between the same individuals shifted over time.
%a visualization of which can be seen in Figure~\ref{fig:org_sci_example}.
% , which we visualize in Figure~\ref{fig:org_sci_example}. 
Such transformations give rise to several important questions in statistical network inference: how to construct useful measures of dissimilarity across networks; how to estimate any such measure of dissimilarity from random network realizations; how to identify loci of change; and how to gauge differences across scales, from nodes to sub-networks to the entire network itself. 
Our goal in this paper is to build a robust methodology to address such phenomena, and to model and infer important characteristics of network time series.
%, and enable inference in this setting.
\begin{figure}[p]
    \centering
    \includegraphics[width=\textwidth]{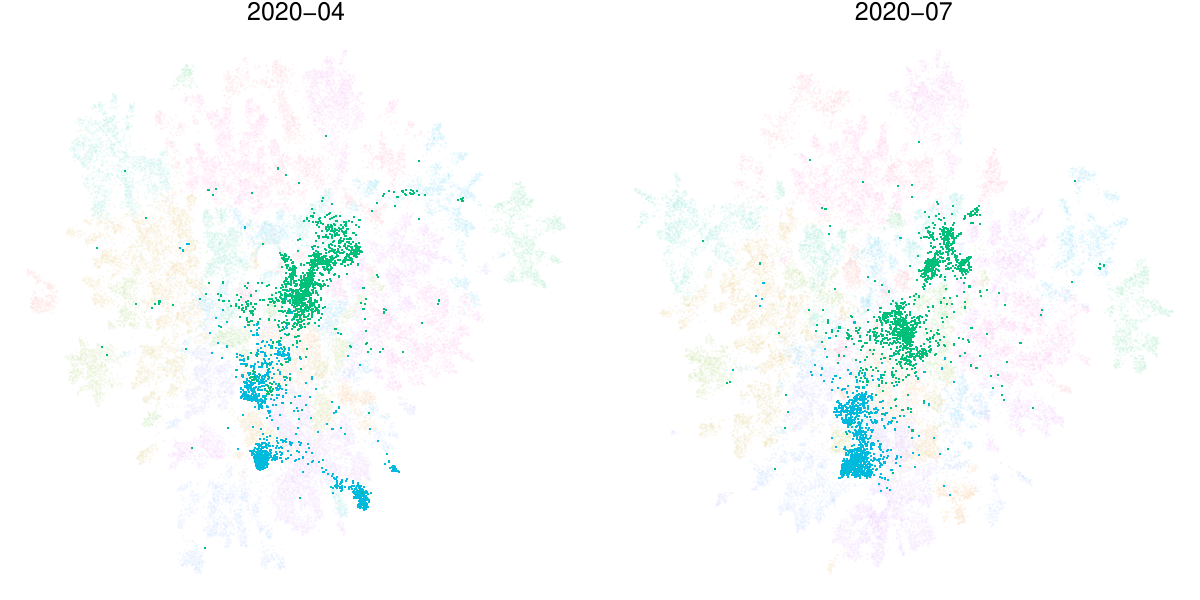}
    \caption{
    Evidence of structural network dissimilarity across time. Visualizations are of anonymized and aggregated Microsoft communications networks in April and July 2020, with two communities highlighted. Nodes represent email accounts, with edges determined by email frequency (edges not shown in these images). Left panel shows an initial Leiden clustering \citep{traag2019louvain} of the nodes into different subcommunities (colored); right panel shows the shift in connectivity structure over this time. Full visualization details are given in Section~\ref{sec:org_sci_data}.
During this time period, our analysis demonstrates that
the network as a whole experiences a structural shock coincident with the pandemic work-from-home order.
However, different subcommunities undergo qualitatively different changes in
their structure, from a connected network that seems to diverge
(green) and a less cohesive one that seems to coalesce (blue).
(See Figure \ref{fig:subcommunities},
 wherein the overall network behavior,
 as well as the two highlighted communities with their different temporal behavior,
 are depicted.) The figures here are two-dimensional renderings of temporal snapshots of a large (n=32277) complex network;
hence, conclusions based on this visualization are notional.
%Our analysis demonstrates quantitatively 
%the claims made here:
%that the green community experiences a constant drift while the blue community experiences a major shock.
}
    \label{fig:org_sci_example}
\end{figure}

To this end, we focus on a class of time series of random networks. We define an intuitive distance between the evolution of certain random variables that govern the behavior of nodes in the networks and prove that this distance can be consistently estimated from the observed networks. When this distance is sufficiently similar to a Euclidean distance, multidimensional scaling extracts a curve in low-dimensional Euclidean space that mirrors the structure of the network dynamics. This permits a visualization of network evolution and identification of change points. Figure~\ref{fig:org_sci_us} is the result of an end-to-end case study using these techniques for a time series of communication networks in a large corporation in the months around the start of pandemic work-from-home protocols: see the dramatic change in both panels beginning in Spring 2020. See Section~\ref{sec:estimation} for the methodology used to generate these figures, and Section~\ref{sec:experiments} for the full details of this experiment.

\begin{figure}[h]
\centering
\includegraphics[width=0.45\textwidth]{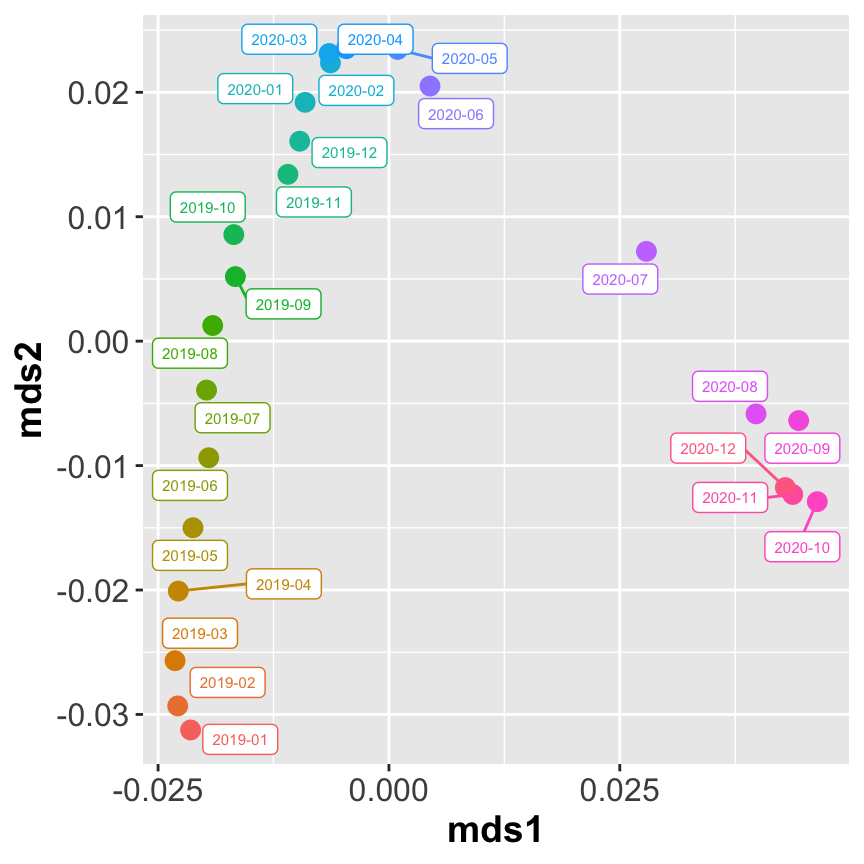}
\includegraphics[width=0.45\textwidth]{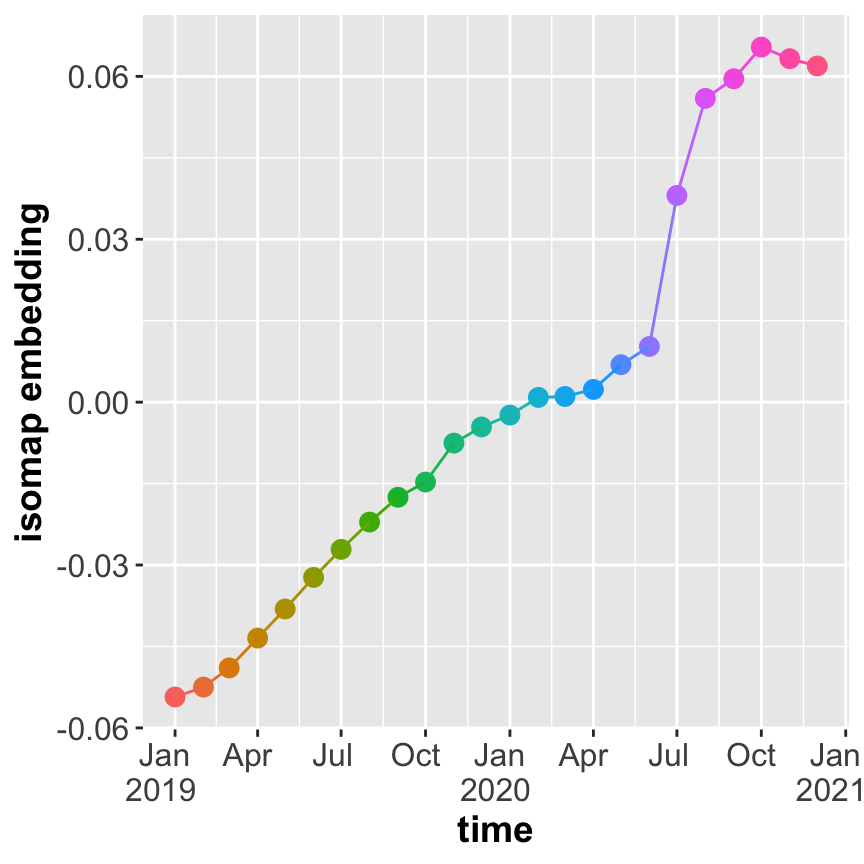}
\caption{Classical multidimensional scaling (CMDS) and ISOMAP embeddings of estimated network dissimilarity identifies changes in network structure. Left plot shows top two dimensions from multidimensional scaling of pairwise network distance matrix for anonymized and aggregated Microsoft communications networks from January 2019 to December 2020, demonstrating that individual networks follow a curve that progresses smoothly until spring 2020 and then exhibits a major shock. Right plot shows associated ISOMAP manifold learning representation and provides clear and concise anomaly and changepoint information. Full details for how these plots were generated appears in Section~\ref{sec:experiments}.}
%(Figure \ref{fig:subcommunities} below shows that the ISOMAP representation of 
%the blue community highlighted in Figure 1 shares the major shock characteristics of the overall network, while the green community highlighted in Figure 1 exhibits no spring 2020 shock.)

\label{fig:org_sci_us}
\end{figure}

Analysis of multiple networks is a key emerging subdiscipline of network inference, with approaches ranging from joint spectral embedding \citep{levin_omni_2017,jones2020multilayer,arroyo2019inference,gallagher2021spectral, jing2021community,pantazis2020importance}, tensor decompositions \citep{zhang2018tensor,lei2020consistent,jing2021community}, least-squares methods \citep{pensky2019dynamic,lei2022bias}, maximum likelihood methods \citep{krivitsky2014separable} and multiscale methods via random walks on graphs \citep{lee2011multiscale}. In \cite{padilla2022change} and \cite{wang2021optimal}, the authors consider changepoint localization for a time series of latent position random graphs \citep{Hoff2002}, a type of independent-edge network in which each node or vertex has an associated latent position that determines its probability of connection with others. The authors establish consistency for localization of a particular kind of changepoint---namely, the case in which the latent positions are all fixed prior to some time point and after which they may be different. Asymptotic properties of these methods depend on particular model assumptions for how the networks evolve over time and relate to one another, and rigorous performance guarantees can be challenging and limited in scope. The underlying geometry of latent spaces affects network structure and evolution as well. In \cite{asta_cls}, the authors consider the impact of different curvature and non-Euclidean properties of latent space geometry on network formation. In \cite{wilkins2022asymptotically}, the authors prove asymptotic results for estimators of underlying latent space curvature.

On the one hand, both single and multiple-network inference problems often have related objectives. For example, if data include multiple network realizations from the same underlying model on the same set of aligned vertices, we may wish to effectively exploit these additional realizations for more accurate estimation of common network parameters---that is, use the replications in a multiple-network setting to refine parameter estimates that govern any single network in the collection. On the other hand, multiple network inference involves statistically distinct questions, such as identifying loci of change {\em across} networks or detecting anomalies in a time series of networks. 

% These latter two inference tasks are fundamentally about how pairs of networks might differ. As such, a natural place to start is a pairwise dissimilarity measure between the networks themselves. Empirical studies on pairwise dissimilarity include CITE Guodong, Mauro, and while these are illuminating, in the current work, we seek to develop a theoretically justified approach for answering such questions about time series of graphs.

Euclidean latent position networks assign to each vertex a typically unobserved vector in some low-dimensional Euclidean space $\mathbb{R}^d$; edges between vertices then arise independently. The probability of an edge between vertex $i$ and vertex $j$ is some fixed function $\kappa$, called the {\em link} function or kernel, of the two associated latent positions for the respective vertices.  
Latent position random graphs have the appealing characteristic of modeling network connections as functions of inherent features of the vertices themselves---these features are encoded in the latent positions---and transforming network inference into the recovery of lower-dimensional structure. More specifically, if we have a series of time-indexed latent position graphs $G_t$ on a common aligned vertex set, then associated to each network is a matrix $\bX_t$ whose rows are the latent vectors of the vertices. Since the edge formation probabilities are a function of pairs of rows of $\bX_t$, the probabilistic evolution of the network time series is completely determined by the evolution of the rows of $\bX_t$. As such, the natural object of study for inference about a time series of latent position graphs are the rows of $\bX_t$. In particular, anomalies or change-points in the time-series of networks correspond to changes in the $\bX_t$ process. For example, a change in a specific network entity is associated to a change in its latent position, which can then be estimated.

The evolution of the rows of $\bX_t$ can be deterministic, as is the case when features of the nodes in a network follow some predictable time-dependent pattern; but it can also be random, as is the case when the actors in a network have underlying preferences that are subject to random shocks. When the latent position vector $X_t(i)$ for some individual vertex $i$ is a random variable, we have, as $t$ varies, a stochastic process. This collection of random variables can be endowed with a metric, which under certain conditions is Euclidean realizable; that is, the random variables at each time have a representation as points in $\RR^c$ for some dimension $c$, where the metric space distances between the random variables are equal to the Euclidean distances between these points (see \cite{borg2005modern} for more on Euclidean realizability of dissimilarity matrices). This allows us to visualize the time evolution of this stochastic process as the image of a map from an interval into $\RR^c$. 

We use this idea to formulate a novel approach to network time series. We demonstrate methods for consistently estimating a Euclidean representation, or \emph{mirror}, of the evolution of the latent position distributions from the observed networks. This mirror can reveal important underlying structure of the network dynamics, as we demonstrate in both simulated and real data, the latter of which is drawn from organizational and communication networks, revealing the change-point corresponding to the start of pandemic work-from-home orders.

\section{Model and Geometric Results}
\label{sec:model}

In order to model the intrinsic characteristics of the entities in our network, we consider \emph{latent position random graphs}, which associate a vector of features in $\RR^d$ to each vertex in the network. The connections between vertices in the network are independent given the latent positions, with connection probabilities depending on the latent position vectors of the two vertices in question. In our notation, $X \in \mathcal{X} \subset \mathbb{R}^d$ or $x \in \mathcal{X} \subset \mathbb{R}^d$, represent column vectors. If such column vectors are arranged as rows in a matrix, we specify this explicitly or we use the transpose to denote the corresponding row vector.

\begin{definition}[Latent Position Graph, Random Dot Product Graph, and Generalized Random Dot Product]
We say that the random graph $G$ with adjacency matrix $\bA\in\RR^{n\times n}$ is a \emph{latent position random graph (LPG)} with latent position matrix $\bX\in\RR^{n\times d}$, whose rows are the transpositions of the column vectors $X^1,\ldots,X^n\in\mathcal{X}\subseteq\RR^d$, and link function $\kappa: \mathcal{X}\times\mathcal{X}\rightarrow [0,1]$, if
$$\PP[\bA|\bX]=\prod_{i<j}\kappa(X^i,X^j)^{A_{i,j}}(1-\kappa(X^i,X^j))^{1-A_{i,j}}.$$
If $\kappa(x,y)=x^\top y$, we say that $G$ is a \emph{random dot product graph (RDPG)} and we call $\bP=\bX \bX^T$ the connection probability matrix. In this case, each $A_{ij}$ is marginally distributed (conditionally on $X^i, X^j$) as Bernoulli($\langle X^i,X^j\rangle$).

As a generalization, suppose $x^1, y^1 \in \mathbb{R}^p$ and $x^2, y^2\in \mathbb{R}^q$, where $p+q=d$. When $\kappa([(x^1)^\top,(x^2)^\top]^\top,[(y^1)^\top,(y^2)^\top]^\top)=(x^1)^\top y^1-(x^2)^\top y^2$, we say that $G$ is a \emph{generalized random dot product graph (GRDPG)} and we call $\bP=\bX I_{p,q} \bX^T$ the generalized edge connection probability matrix, where $I_{p,q}=I_p\oplus(-I_q)$.
\end{definition}

\begin{remark}[Orthogonal nonidentifiability in RDPGs] \label{rem:nonid}
Note that if $\bX \in \R^{n \times d}$ is a matrix of latent positions
and $\bW \in \R^{d \times d}$ is orthogonal,
$\bX$ and $\bX\bW$ give rise to the same distribution over graphs.
Thus, the RDPG model has a nonidentifiability up to orthogonal transformation. Analogously, the GRDPG model has a nonidentifiability up to indefinite orthogonal transformations.
\end{remark}
Since we wish to model randomness in the underlying features of each vertex, we will consider latent positions that are themselves random variables defined on a probability space $(\Omega, \mathcal{F}, \PP)$. For a particular sample point $\omega \in \Omega$, let $X(t,\omega)\in\RR^d$ be the realization of the associated latent position for this vertex at time $t \in [0, T]$. On the one hand, for fixed $\omega$, as $t$ varies, $X(t,\omega), 0 \leq t \leq T$ is the realized trajectory of a $d$-dimensional stochastic process. On the other hand, for a given time $t$, the random variable $X(t, \cdot)$ represents the constellation of possible latent positions at this time. In order for the inner product to be a well-defined link function, we require that the distribution of $X(t,\cdot)$ follow an \emph{inner-product distribution}:
\begin{definition}
\label{def:innerprod}
Let $F$ be a probability distribution on $\R^d$.
We say that $F$ is a
\emph{$d$-dimensional inner product distribution}
if $0 \le x^{\top} y \le 1$ for all $\bx,\by \in \supp F$. We will suppose throughout this work that for a $d$-dimensional inner product distribution $F$ and $X\sim F$, $\EE[XX^\top]$ has rank $d$.
\end{definition}
% For convenience, we will sometimes consider RDPGs with latent positions that do not satisfy $x^\top y\in[0,1]$ for all latent positions $x,y$; in this case the link function should return $0$ or $1$ rather than $x^\top y$. 

We wish to quantify the difference between the random vectors $X(t,\cdot)$ and $X(t',\cdot)$. Suppose that the graphs come from an RDPG or GRDPG model, where at each time $t$, the latent positions of each graph vertex are drawn independently from a common inner product latent position distribution $F_t$. Because $X(t, \cdot)$ is a latent position, we necessarily have $X(t, \cdot) \in L^2(\Omega)$; for notational simplicity, we will use $X(t, \cdot)$ and $X_t$ interchangeably. We define a norm, which we call the {\em maximum directional variation norm}, on this space of random variables; this norm leads to a natural metric $d_{MV}(X_t, X_{t'})$, both of which are described below. In the definition below, and throughout the paper, we use $\|\cdot\|$ to denote the Euclidean norm in $\RR^d$, $\|\cdot\|_2$ to denote the spectral norm of a matrix, and $\|\cdot\|_F$ to denote the Frobenius norm of a matrix.

\begin{definition}[Maximum directional variation norm and metric]\label{def:dMV}
For a random vector $X\in L^2(\Omega)$, we define $$\|X\|_{MV}=\max_u \EE[\langle X,u\rangle^2]^{1/2}=\|\EE[XX^\top]\|_2^{1/2},$$ where the maximization is over $u\in\RR^d$ with $\|u\|=1$.
We define an associated metric $d_{MV}$ by minimizing the norm of the difference between the random variables $X_t, X_{t'}$ over all orthogonal transformations, which aligns these distributions.
\begin{equation}
\label{eq:dr-dist}
d_{MV}(X_t,X_{t'})=\min_{W}\|X_t-WX_{t'}\|_{MV}=\min_{W} \left\|\EE[(X_t-WX_{t'})(X_t-WX_{t'})^\top]\right\|_2^{1/2},
\end{equation}
where the matrix norm on the right hand side is the spectral norm. Given a map $\varphi:[0,T]\rightarrow L^2(\Omega)$ that assigns time points $t$ to random variables $\varphi(t)=X_t$, we may write $d_{MV}(\varphi(t),\varphi(t')):=d_{MV}(X_t,X_{t'}).$
\end{definition}
The minimization in Equation~\ref{eq:dr-dist} is a variant of the classical Procrustes alignment problem, so we may refer to the latent positions after this rotation as ``Procrustes-aligned." 
% If $X$ has mean zero, the $\|X\|_{MV}$ considers the square of spectral norm of its covariance matrix; that is, the norm $\|\cdot\|_{MV}$ gives the maximal directional \emph{variation} when $X$ is centered. In the cases of interest, we wish to capture features of the variance of the {\em drift} in the latent position, $X_t-WX_{t'}$: this is the origin of the name for this metric and its associated norm.

\begin{remark}
If $X$ has mean zero, the $\|X\|_{MV}$ considers the square of spectral norm of its covariance matrix; that is, the norm $\|\cdot\|_{MV}$ gives the maximal directional \emph{variation} when $X$ is centered. In the cases of interest, we wish to capture features of the variance of the {\em drift} in the latent position, $X_t-WX_{t'}$: this is the origin of the name for this metric and its associated norm.
The metric $d_{MV}$ is not properly a metric on $L^2(\Omega)$, since if $X=WY$ a.s. for some orthogonal matrix $W$, then $d_{MV}(X,Y)=0$. However, if we consider the equivalence relation defined by $X\sim Y$ whenever $X=WY$ for some orthogonal matrix, this \emph{is} a metric on the corresponding set of equivalence classes. This means that we are able to absorb the non-identifiability from the original parameterization, obtaining a new parameter space with a metric space structure where the underlying distribution is identifiable.
\end{remark}

One of the central contributions of this paper is that the $d_{MV}$ metric captures important features of the time-varying distributions $F_t$. To describe a family of networks indexed by 
time, each of which is generated by a matrix of latent positions that are themselves random, we consider a {\em latent position stochastic process}.
\begin{definition}[Latent position process]\label{def:latent_pos_proc}
Let $\mathcal{F}_t, 0 \leq t \leq T$ be a filtration of $\mathcal{F}$. A {\em latent position process} $\varphi(t)$ is an $\mathcal{F}_t$-adapted map $\varphi:[0,T]\rightarrow (L^2(\Omega),d_{MV})$ such that for each $t\in [0,T]$, $\varphi(t)=X(t,\cdot)$ has an inner product distribution. We say that a latent position process is {\em nonbacktracking} if $\varphi(t)=\varphi(t')$ implies $\varphi(s)=\varphi(t)$ for all $s \in [t,t']$.
\end{definition}

Once we have the latent position stochastic process, we can construct a time series of latent position random networks whose vertices have independent, identically distributed latent positions given by $\varphi(t)$.

\begin{definition}[Time Series of LPGs]
\label{def:tsg}
Let $\varphi$ be a latent position process, and fix a given number of vertices $n$ and collection of times $\mathcal{T}\subseteq[0,T]$. We draw an i.i.d.\ sample $\omega_j\in \Omega$ for $1\leq j\leq n$, and obtain the latent position matrices $\bX_{t}\in\RR^{n\times d}$ for $t\in\mathcal{T}$ by appending the rows $X(t,\omega_j)$, $1\leq j\leq n$. The \emph{time series of LPGs (TSG)} $\{G_t: t\in\mathcal{T}\}$ are conditionally independent LPGs with latent position matrices $\bX_{t}, t\in\mathcal{T}$.
\end{definition}

We emphasize that each vertex in the TSG corresponds to a single $\omega\in\Omega$, which induces dependence between the latent positions for that vertex across time points, but the latent position trajectories of any two distinct vertices are independent of one another across all times. Since these trajectories form an i.i.d.\ sample from the latent position process, it is natural to measure their evolution over time using the metric on the corresponding random variables, namely $d_{MV}(X_t,X_{t'})$. In the definition of this distance, the expectation is over $\omega\in\Omega$, which means that it depends on the joint distribution of $X_t$ and $X_{t'}$. In particular, $d_{MV}$ depends on more than just the marginal distributions of the random vectors $X_t$ and $X_{t'}$ individually, but takes into account their dependence inherited from the latent position process $\varphi$. 

A key question is whether the image $\varphi([0,T])$ has useful geometric structure when equipped with the metric $d_{MV}$. It turns out that, under mild conditions, this image is a manifold. In addition, the map $\varphi$ admits a Euclidean analogue, called a {\em mirror}, which is a finite-dimensional curve that retains important signal from the generating stochastic process for the network time series. To make this precise, we define the notions of Euclidean realizability and approximate Euclidean realizability, below, and provide several examples of latent position processes that satisfy these requirements.

\begin{definition}[Notions of Euclidean realizability]
Let $\varphi$ be a latent position process. 
%We say that this is \emph{Lipschitz Euclidean realizable in dimension $c$} with \emph{mirror} $\psi$ if there exists a Lipschitz continuous curve $\psi:[0,T]\rightarrow\RR^c$ such that $$d_{MV}(\varphi(t),\varphi(t'))=\|\psi(t)-\psi(t')\|\quad\text{ for all }t,t'\in[0,T].$$

We say that $\varphi$ is \emph{approximately (Lipschitz) Euclidean $c$-realizable} with \emph{mirror} $\psi$ and realizability constant $C>0$ if there exists a Lipschitz continuous curve $\psi:[0,T]\rightarrow\RR^c$ such that
$$\left|d_{MV}(\varphi(t),\varphi(t'))-\|\psi(t)-\psi(t')\|\right|\leq C|t-t'|\text{ for all }t,t'\in[0,T].$$

For a fixed $\alpha\in(0,1)$, we say that $\varphi$ is \emph{approximately $\alpha$-H\"{o}lder Euclidean $c$-realizable} if $\psi$ is $\alpha$-H\"{o}lder continuous, and there is some $C>0$ such that 
$$\left|d_{MV}(\varphi(t),\varphi(t'))-\|\psi(t)-\psi(t')\|\right|\leq C|t-t'|^{\alpha}\text{ for all }t,t'\in[0,T].$$ Rather than $c$-realizable, we may simply say \emph{realizable} if there is some $c$ for which this holds; we simply say that $\varphi$ is H\"{o}lder Euclidean realizable if the condition holds for some $\alpha\in(0,1]$.
\end{definition}

\begin{remark} If there exists a Lipschitz curve $\psi$ in $\mathbb{R}^c$ for which 
$$d_{MV}(\varphi(t), \varphi(t'))=\|\psi(t)- \psi(t')\|$$
we say the latent position process is {\em exactly Euclidean realizable}. While this is seldom the case for most interesting latent position processes, it can be instructive to consider what this implies: that pairwise $d_{MV}$ distances between the latent position process at $t$ and $t'$ coincide exactly with Euclidean distances along the curve $\psi$ at $t$ and $t'$. Hence the term {\em mirror}, a Euclidean-space curve that replicates (with some distortion in the approximately realizable case) the time-varying $d_{MV}$ distance.
For useful intuition, consider a one-dimensional Brownian motion $B_t$. While this is not a latent position process, its covariance operator $R(s,t)=\mathbb{E}[(B_t-B_s)^2]$ is exactly 
$\mathbb{E}[(B_t-B_s)^2]=(t-s)$, corresponding to the distance between points along the line $x(t)=t$ between $t$ and $s$.
\end{remark}

In practice, the latent position process is unobserved, so it is unclear whether the Euclidean realizability condition holds. However, we show that the $d_{MV}$ distance can be consistently estimated, so the question of realizability may be resolved at least in part by inspection of the scree plot of the estimated distance matrix. We remark further on this point after Theorem~\ref{thm:approximation}.

Note that if $\varphi$ is approximately $c$-realizable, it is $c'$-realizable for any $c'>c$, and a trade-off exists between the choice of dimension $c$ and the accuracy of the approximation, as measured by $C$ and $\alpha$. The realizability dimension $c$ can be interpreted as a choice in a dimension reduction procedure. Namely, the dimension $c$ corresponds to a curve $\psi$ in $\mathbb{R}^c$, along which pairwise Euclidean distances locally approximate those of the maximum variational distances along the latent position process. 

As such, none of $c$, $C$, or $\psi$, as defined above, need be unique. This leads naturally to the question of an ``optimal" mirror---that is, one that best captures, in Euclidean space, the salient features of the $d_{MV}$ distance. To make this precise, suppose $\varphi$ is  %approximately $\alpha$-H\"older Euclidean $c$-realizable 
a latent position process. 
%with realizability constant $C$. 
%Define
%$$\mathcal{S}(\alpha, c, C)=\{\psi:[0,T] \rightarrow \mathbb{R}^c: \psi \textrm{ is H\"older } \alpha \textrm{ with realizability constant } C\},$$ and 
For any associated mirror $\psi$, consider the functional $\mathcal{L}(\psi)$ given by
\begin{equation}
\label{eq:functional_L}
\mathcal{L}(\psi)=\int_0^T \int_0^T \left|d_{MV}^2(\varphi(s), \varphi(t))-\|\psi(s)-\psi(t)\|^2\right|^2\,\mathrm{d}s\,\mathrm{d}t
\end{equation}
%and the associated variational problem of its minimization:
%$$\min_{\psi \in S(\alpha, c, C)} \mathcal{L}(\psi)$$
As we show below, there exists a solution to the variational problem of minimizing this functional over the class of mirrors $\mathcal{S}(c, \alpha, C)$ that are $\alpha$-H\"older with realizability constant $C$, and satisfy $\int_0^T \psi(t)\,\mathrm{d}t=0$. We call this minimizer an {\em optimal} mirror for this $\alpha, c,$ and $C$. While any mirror satisfying the realizability constraints estimates the $d_{MV}$ distance well locally, a minimizer of this functional also estimates the $d_{MV}$ distance well in a global sense.
%However, for a given choice of dimension $c$ and $\alpha$, there exists a smallest possible realizability constant $C$.

\begin{thm}[Existence of Optimal Mirrors]
\label{thm:optimal}
Let $\varphi$ be a latent position process, which is approximately $\alpha$-H\"{o}lder Euclidean $c$-realizable with realizability constant $C$. Let $\mathcal{S}(c, \alpha, C)$ be the class of mirrors defined above. Then there exists a solution to the variational problem
\begin{equation}\label{eq:var_prob_L}
\inf_{\psi \in \mathcal{S}(c,\alpha,C)} \mathcal{L}(\psi)
\end{equation}
\end{thm}

\begin{thm}[Uniqueness of Optimal Mirrors]
\label{thm:unique}
If $\varphi$ is a latent position process which is exactly $\alpha$-H\"{o}lder Euclidean $c$-realizable, the solution to the variational problem in Eq. \eqref{eq:var_prob_L} is unique up to orthogonal transformations.
\end{thm}

%\begin{thm}[Smallest realizability constant for a given $c$ and $\alpha$]
% \label{thm:optimal}
% Let $\varphi$ be a latent position process which is approximately Euclidean $c$-realizable for a given $\alpha$. There exists a smallest possible $C$ for this choice of $\alpha$ for which there exists a Euclidean mirror $\psi$. %\textcolor{blue}{Moreover, for any $\psi_1$ and $\psi_2$ satisfying these conditions for these optimal $\alpha,C$, there is a vector $v$ and rotation matrix $R$ such that $\psi_2(t)=R\psi_1(t)+v$ for all $t\in[0,T]$.}
%\end{thm}

%{\color{red} Add definition Frobenius-minimizing continuous function A THE mirror here}
%{\color{red} Add Theorem that minimizer exists in the set $S(\alpha, c, C)$.}

% removed just for space arrrggghhh
In addition to the existence of optimal mirrors, an approximate Euclidean realizable latent position process $\varphi$ has the property that its image is a manifold.

\begin{thm}[Manifold properties of a nonbacktracking latent position process]
\label{thm:manifold}
Let $\varphi$ be a nonbacktracking latent position process which is approximately Euclidean realizable. Then $\mathcal{M}=\varphi([0,T])$ is homeomorphic to an interval $[0,I]$. In particular, it is a topological 1-manifold with boundary. If $\varphi$ is injective and approximately $\alpha$-H\"{o}lder Euclidean realizable, the same conclusion holds.
\end{thm}

If we suppose that the trajectories of $\varphi$ satisfy a certain degree of smoothness, it turns out that the map $\varphi$ into the space of random variables equipped with the $d_{MV}$ metric also has this degree of smoothness.

\begin{thm}(Smooth trajectories and smooth latent position processes)
\label{thm:smoothtraj}
Suppose $X(\cdot,\omega):[0,T]\rightarrow \RR^d$ is $\alpha$-H\"{o}lder continuous with some $\alpha\in(0,1]$ for almost every $\omega\in\Omega$, such that $$\|X(t,\omega)-X(s,\omega)\|\leq L(\omega)|t-s|^{\alpha},$$ where the random variable $L\in L^2(\Omega).$ Let $\mathcal{M}=\varphi([0,T])$. Then $\varphi:[0,T]\rightarrow (\mathcal{M},d_{MV})$ is H\"{o}lder continuous with this same $\alpha$.
\end{thm}

\begin{remark}
\label{rem:CMDS_and_Euclidean_Realizability}
In the above definitions of realizability, regularity conditions are imposed on $\psi$, which takes values in $\RR^c$, rather than on $\varphi$, which gives random variables as output. Moreover, $\psi$ is the Euclidean realization of the manifold $\varphi([0,T])$ in the space of random variables; this as an approximately distance-preserving representation of those random variables, each of which captures the full state of the system with all of the given entities at any time $t$. As we show, estimates of this Euclidean mirror, derived from observations of graph connectivity structure at a collection of time points, can recover important features of the time-varying latent positions.
\end{remark}

There are several natural classes of latent position processes that
are
approximately Lipschitz or $\alpha$-H\"{o}lder Euclidean realizable. The next theorem demonstrates approximate $\alpha$-H\"{o}lder Euclidean realizability for any latent position process expressible as the sum of a deterministic drift and a martingale term whose increments have well-controlled variance.
\begin{thm}[Approximate Holder realizability of variance-controlled martingale-plus-drift processes]
\label{thm:martingale}
Suppose $M_t$ is an $\mathcal{F}_t$-martingale with respect to the filtration $\{\mathcal{F}_t: 0\leq t\leq T\}$, and suppose $\gamma:[0,T]\rightarrow\RR^d$ is Lipschitz continuous. Let $\varphi(t)=\gamma(t)+M_t$. Then
$$d_{MV}(X_t,X_s)^2\leq \|\mathrm{Cov}(M_t-M_s)\|_2+ \|\gamma(t)-\gamma(s)\|^2.$$
%When $M_t$ satisfies $\|\mathrm{Cov}(M_t-M_s)\|_2\leq C(t-s)$, then $\varphi$ is approximately $\alpha$-Holder Euclidean realizable with $\alpha=1/2$, with $\psi=\gamma$.
When $M_t$ satisfies $\|\mathrm{Cov}(M_t-M_s)\|_2\leq C(t-s)$, and $\gamma(t)=a(t)v$ for some $v\in\RR^d$ and Lipschitz continuous $a:[0,T]\rightarrow\RR$, then $\varphi$ is approximately $\alpha$-H\"{o}lder Euclidean realizable with $\alpha=1/2$ and $c=1$.
\end{thm}

\begin{example}
\label{ex:BM}
Consider $X_t=\gamma(t)+B_t$, where $B_t$ is a $d$-dimensional Brownian motion, and $\gamma:[0,T] \rightarrow \mathbb{R}^d$ is a Lipschitz continuous function of the form $\gamma(t)=a(t)v$. The $\varphi(t)$ is approximately $\alpha$-H\"{o}lder Euclidean realizable, with $\alpha=1/2$, and the Euclidean mirror $\psi(t)$ is $\psi(t)=a(t)\|v\|$, so $c=1$. In Section \ref{sec:drift-plus-noise_LPP}, we provide simulations of a network time series with this latent position process and show that our estimated mirror matches $\psi$ well for a network of 2000 nodes.
\end{example}

\begin{example}
\label{ex:integratedBM}
Consider $X_t= \gamma(t)+I_t$, where $I_t=\int_{0}^{t} B_s\,\mathrm{d}s$, $B_s$ is a $d$-dimensional Brownian motion, and $\gamma(t)=(at+b)v$ is a function describing the mean of $X_t$ over time, with $a,b\in\RR, v\in\RR^d$.
Then each sample path of $X_t$ is continuously differentiable in $t$, and $\EE[X_t]=\gamma(t)$ is as well.
If $\{\mathcal{F}_t: 0 \leq t \leq T\}$ is the canonical filtration generated by Brownian motion, then $I(t)$ is {\em not} an $\mathcal{F}_t$-martingale. Then $$d_{MV}(X_t,X_{t'})^2= a^2(t-t')^2\|v\|^2+\sigma^2[(t-t')^2(t+2t')/3],$$ so $\varphi$ is approximately Euclidean realizable with $\psi(t)=\sqrt{a^2\|v\|^2+\sigma^2T}t$, so again $c=1$.
\end{example}

The latent positions for the vertices in our network are not typically observed---instead, we only see the connectivity between the nodes in the network, from which a given realization of the latent positions can, under certain model assumptions, be accurately estimated. In order to compare the networks at times $t$ and $t'$, we can consider estimates of the networks' latent positions at these two times as noisy observations from the joint distribution of $(X_t,X_{t'})$, and deploy these estimates in an approximation of the distances $d_{MV}(X_t,X_{t'})$. Using these approximate distances, we can then estimate the curve $\psi(t)$, giving a visualization for the evolution of \emph{all} of the latent positions in the random graphs over time.

Suppose that $G$ is a random dot product graph with latent position matrix $\bX$, where the rows of $\bX$ are independent, identically distributed draws from a latent position distribution $F$ on $\mathbb{R}^d$. Let $\bA$ be the adjacency matrix for this graph. As shown in \cite{STFP-2011}, a spectral decomposition of the adjacency matrix yields consistent estimates for the underlying matrix of latent positions. We introduce the following definition.
\begin{definition}[Adjacency Spectral Embedding]
Given an adjacency matrix $\bA$, we define the \emph{adjacency spectral embedding (ASE)} with dimension $d$ as $\hat{\bX}=\hat{\bU}\hat{\bS}^{1/2}$, where $\hat{\bU}\in\RR^{n\times d}$ is the matrix of $d$ top eigenvectors of $\bA$ and $\hat{\bS}\in\RR^{d\times d}$ is the diagonal matrix with the $d$ largest eigenvalues of $\bA$ on the diagonal.
\end{definition}
As we show in the next section, we will use the ASE of the observed adjacency matrices in our TSG to estimate the $d_{MV}$ distance between latent position random variables over time, and in turn, to estimate the Euclidean mirror, which records important underlying structure for the time series of networks. 
\section{Statistical Estimation of Euclidean Mirrors}
\label{sec:estimation}

Given a finite sample from a time series of graphs with approximately $\alpha$-H\"{o}lder Euclidean realizable latent position process $\varphi$, our goal is to estimate a finite-sample analogue of an optimal Euclidean mirror $\psi$. The distances $d_{MV}(\varphi(t),\varphi(s))$ can be used to recover a version of the mirror at these sampled times (up to rigid transformations) from classical multidimensional scaling (CMDS). As such, the crucial estimation problem is one of accurately estimating the distances $d_{MV}(\varphi(t),\varphi(s))$. To this end, we define the {\em estimated pairwise distances} between any two such $n \times d$ latent position matrices $\hat{\bX}_{t}$ and $\hat{\bX}_{s}$ as follows:
\begin{equation}\label{eq:estimated_d}
\hat{d}_{MV}(\hat{\bX}_{t},\hat{\bX}_{s}):=\min_{W\in\mathcal{O}^{d\times d}} \frac{1}{\sqrt{n}}\|\hat{\bX}_{t}-\hat{\bX}_{s}W\|_2,
\end{equation}
where $\mathcal{O}^{d\times d}$ is the set of real orthogonal matrices of order $d$, and $\|\cdot\|_2$ denotes the spectral norm. Note that when $\mathbf{U},\mathbf{V}\in M_{n,d}(\RR)$ have orthonormal columns, we have the following well-known relations between $\hat{d}_{MV}$ and the spectral norm of their $\sin\Theta$ matrix:
$$ \|\sin \Theta(\mathbf{U},\mathbf{V})\|_2\leq \sqrt{n}\hat{d}_{MV}(\mathbf{U},\mathbf{V})\leq \sqrt{2}\|\sin\Theta(\mathbf{U},\mathbf{V})\|_2. $$
Our central result is that, when our networks have a sufficiently large
number of vertices $n$, $\hat{d}_{MV}$ provides a consistent estimate of $d_{MV}$.
\begin{thm}
\label{thm:approximation}
Suppose $\varphi$ is a LPP such that $\varphi(t)$ takes values in $\RR^d$ for all $t\in[0,T]$. Let $\{G_t\}_{t\in\mathcal{T}}$ be a time series of LPGs whose latent positions follow this LPP. For each $t\in\mathcal{T}$, let $\hat{\bX}_t \in \RR^{n\times d}$ be the matrix of estimated latent positions from the ASE of each graph $G_t$. Then for all $s,t\in\mathcal{T}$, with overwhelming probability as $n\rightarrow\infty$, 
$$
\left|\hat{d}_{MV}(\hat{\bX}_{t},\hat{\bX}_{s})^2-d_{MV}(\varphi(t),\varphi(s))^2\right|\leq \frac{\log(n)}{\sqrt{n}}.
$$
\end{thm}

The functional $\mathcal{L}(\psi)$ in Equation~\ref{eq:functional_L} requires information of $d_{MV}(\varphi(t),\varphi(s))$ for all $t,s\in[0,T]$. In the finite-sample case, however, we only have a fixed, finite set of time points $\mathcal{T}=\{t_i\}_{i=1}^{m}\subseteq[0,T]$, with $t_i < t_{i+1}$ for all $i$. To address  finite-sample estimation of an analogue of an optimal mirror, we introduce the functional $\widehat{\mathcal{L}}$
%. For a finite set of time points $\mathcal{T}=\{t_i\}_{i=1}^{m}\subseteq[0,T]$, 
defined on sets of size $m$ of vectors in $\RR^c$:
\begin{equation}
    \label{eq:functionalhat}
    \widehat{\mathcal{L}}(\{v_i\}_{i=1}^{m})=\sum_{i,j=1}^{m} \left|d_{MV}(\varphi(t_i),\varphi(t_j))^2-\|v_{i}-v_{j}\|^2\right|^2\;\Delta t_i\,\Delta t_j.
\end{equation}

The time steps are defined as $\Delta t_i:=t_i-t_{i-1}$, with $t_0= 0$. Note that the time steps $\Delta t_i$ need not be constant for Equation~\ref{eq:functionalhat}, allowing us to consider real-data settings in which the network observations may not be equally spaced in time. Suppose we know the true matrix $D_{\varphi}$ of pairwise distances whose $i,j$th entry is $d_{MV}(\varphi(t_i), \varphi(t_j))$. When the $\Delta t_i$ are equal and the process is exactly Euclidean realizable, classical multidimensional scaling (CMDS) applied to this matrix yields a collection of vectors $\{\psi(t_i): 1 \leq i \leq m\}$, unique up to rotation, that minimizes $\widehat{\mathcal{L}}(\{v_i\}_{i=1}^{m})$ (see additional details after Corollary~\ref{cor:approximationmatrices}). We call this the finite-sample mirror for the latent position process $\varphi$. When the $\Delta t_i$ are not all equal, we suggest obtaining an initial solution for Equation~(\ref{eq:functionalhat}) via CMDS, obtaining $V_0 \in \RR^{m\times c}$, with rows $v_i\in \RR^c$. Then one can employ a general-purpose nonlinear solver to Equation~(\ref{eq:functionalhat}) with the initial point $V_0$.

%{\color{red} With Frobenius norm minimizer as given before, note that $D_{\psi}$ can be defined. Consider an analogous finite-time-sampled version of the same minimization problem, and call the solution to this $\psi$ . We can obtain this via CMDS of $D_{\phi}$. Thus, CMDS of the $\hat{d}_{MV}$ is close to CMDS of $D_{\phi}$, but this doesn't necessarily produce the true minimizing $\psi$ at times $t_1, \cdots, t_k$. Our consistently result is that we correctly estimate the minimizer of the finite-time-sampled double-sum objective, which is "the mirror" for this choice of $c$, up to rotations, in the finite-time case.}

%In order to characterize the evolution of time series of networks, consider a finite subset of times $\mathcal{T}\subseteq[0,T]$ with $|\mathcal{T}|=m$, and 
Having defined the dissimilarity matrices 
$$
\mathcal{D}_{\varphi}=[d_{MV}(\varphi(s),\varphi(t))]_{s,t\in\mathcal{T}},\quad \mathcal{D}_{\psi}=[\|\psi(s)-\psi(t)\|_2]_{s,t\in\mathcal{T}},%\quad\textrm{ and }\quad \mathcal{D}_{\hat{\psi}}=[\hat{d}_{MV}(\hat{\bX}_{s},\hat{\bX}_t)]_{s,t\in\mathcal{T}},
$$
where $\{\psi(t)\}_{t\in\mathcal{T}}$ achieves the minimum value of $\widehat{\mathcal{L}}$, we note that the first records the pairwise distances between the latent position process at times $t$ and $s$; the second records the differences between the finite-sample optimal Euclidean mirror at these times. Of course, the true distances are not observed, and must be estimated. The estimates for these quantities are then 
$$
\widehat{\mathcal{D}}_{\varphi}=[\hat{d}_{MV}(\hat{\bX}_{s},\hat{\bX}_t)]_{s,t\in\mathcal{T}},\quad \mathcal{D}_{\hat{\psi}}=[\|\hat{\psi}(s)-\hat{\psi}(t)\|_2]_{s,t\in\mathcal{T}},
$$
where $\{\hat{\psi}(t)\}_{t\in\mathcal{T}}$ is the output of CMDS applied to the matrix $\widehat{\mathcal{D}}_{\varphi}$. This means our full mirror estimation procedure is as follows:

\begin{algorithm}[h!]
    \caption{Mirror estimation}\label{alg:Mirror_estimation}
    \begin{algorithmic}[1]
        \State {{\bf Input}: graph adjacency matrices $A_{t_1}, \cdots, A_{t_{m}}$, embedding dimensions $d$ and $c$.}
			%\For{$i \in \{1, \cdots, T\}$,}
	    \State Compute Adjacency Spectral Embedding of $A_{t_i}$ to obtain $\hat{\bf X}_{t_1}, \cdots, \hat{\bf X}_{t_{m}}\in M_{n,d}(\RR).$
    
	    \State For $i, j \in \{1, \cdots, m\}$ compute matrix entry $\hat{\mathcal{D}}_{\varphi}(i,j)=\hat{d}_{MV}(\hat{\bf X}_{t_i}, \hat{\bf X}_{t_j})$.
    
        \State Apply CMDS to $\hat{\mathcal{D}}_{\varphi}$ to yield $\hat{\psi}(t_i)\in\RR^c, 1 \leq i \leq m$. The $\hat{\psi}(t_i)$ are defined as the rows of $\hat{U}\hat{S}^{1/2}$, where $\hat{U}\hat{S}\hat{U}^T$ is the closest rank-$c$ matrix to $-\frac{1}{2}P\hat{\mathcal{D}}_{\varphi}^{(2)}P$ with respect to the Frobenius norm, and $P=I-J/m$, for $J$ the $m\times m$ matrix of all ones.
	    \State {\bf Output}: Return mirror estimates $\hat{\psi}(t_1), \cdots, \hat{\psi}(t_{m})$.
    \end{algorithmic}
\end{algorithm}

Suppose $A^{(2)}$ is the matrix of squared entries of $A$. 
Theorem \ref{thm:approximation} then guarantees that the square of each entrywise difference between $\hat{\mathcal{D}}^{(2)}_{\varphi}$ and $\mathcal{D}_{\hat{\psi}}^{(2)}$ is bounded, with high  probability, by $\log^2(n)/n$. Since both $\hat{\mathcal{D}}^{(2)}_{\varphi}$ and $\mathcal{D}_{\hat{\psi}}^{(2)}$ are $m \times m$ matrices, we immediately derive the following corollary:
\begin{corollary}
\label{cor:approximationmatrices}
Suppose the setting of Theorem~\ref{thm:approximation}. For fixed $m$, with overwhelming probability,
$$
\|\widehat{\mathcal{D}}_{\varphi}^{(2)}-\mathcal{D}_{\varphi}^{(2)}\|_F \leq \frac{m \log(n)}{\sqrt{n}}.
$$
\end{corollary}

We recall that CMDS computes the scaled eigenvectors of the matrix $-\frac{1}{2}PA^{(2)}P$, where $P=I-J_{m}/m$ is a projection matrix, $A^{(2)}$ is a matrix of squared distances, and $J_{m}$ is the $m\times m$ matrix of all ones. This matrix may be written as $USU^T$ for some $m\times c$ matrix $U$ with orthonormal columns and diagonal matrix $S$. This means that $\psi(t_i)$, the value at $t_i$ of the finite-sample optimal Euclidean mirror associated to $\varphi$, is simply the $i$th row of the matrix $US^{1/2}$. We will analogously denote the $i$th row of $\hat{U}\hat{S}^{1/2}$, the output of CMDS applied to $\widehat{\mathcal{D}}_{\varphi}$, by $\hat{\psi}(t_i)$. 

\begin{remark}
\label{rmk:choosec}
In practice, where $c$ is unknown, the selection of the mirror dimension $c$ is a model selection problem. However, in light of Corollary~\ref{cor:approximationmatrices} and the Hoffman-Wielandt inequality, we see that the eigenvalues of the estimated projected distance matrix $-\frac{1}{2}P\widehat{\mathcal{D}}_{\varphi}^{(2)}P$ approximate those of the theoretical one, $-\frac{1}{2}P\mathcal{D}_{\varphi}^{(2)}P$. As such, the singular values of $-\frac{1}{2}P\widehat{\mathcal{D}}_{\varphi}^{(2)}P$ are consistent estimates for the true singular values, meaning that the correct choice of $c$ will be revealed for large networks.
\end{remark}

It turns out that for large networks, the invariant subspace associated to $\widehat{\mathcal{D}}^{(2)}_{\varphi}$ corresponding to its largest eigenvalues is an accurate approximation to the corresponding subspace of $\mathcal{D}^{(2)}_{\varphi}$, which matches that of $\mathcal{D}^{(2)}_{\psi}$ when we have approximate Euclidean realizability. This suggests that applying CMDS to the estimated dissimilarity matrix $\widehat{\mathcal{D}}_{\varphi}$ can recover the finite-sample optimal mirror $\psi(t)$ up to a rotation: in other words, $\hat{\psi}(t)\approx R\psi(t)$ for some real orthogonal matrix $R$ and all $t\in\mathcal{T}$. 

%Our principal result is that in the case of Euclidean realizable latent position processes, this procedure yields consistent estimates of the mirror.
\begin{thm}
\label{thm:cmds}

Suppose $\varphi$ is approximately Euclidean $c$-realizable and let $m$ be a given positive integer. Let $\hat{U},U\in\RR^{m\times c}$ be the top $c$ eigenvectors, and $\hat{S},S\in\RR^{c\times c}$ be the diagonal matrices with diagonal entries equal to the top $c$ eigenvalues of $\hat{E}_{\varphi}=-\frac{1}{2}P\widehat{\mathcal{D}}_{\varphi}^{(2)} P^\top$ and $E_{\varphi}=-\frac{1}{2}P \mathcal{D}_{\varphi}^{(2)} P^\top$, respectively, where $P=I-J_{m}/m$, $J_{m}$ is the all-ones matrix of order $m$. Suppose $S_{i,i}>0$ for $1\leq i\leq c$. Then with overwhelming probability, there is a real orthogonal matrix $R\in\mathcal{O}^{c\times c}$ such that 
$$\|\hat{U}-UR\|_F\leq \frac{2^{3/2}}{\lambda_c(E_{\varphi})}\left(\frac{m\log(n)}{\sqrt{n}}+\left(\sum_{i=c+1}^{m}\lambda^2_i\left(E_{\varphi}\right)\right)^{1/2}\right).$$ Call this upper bound $B=B(n,m,c)$. The spectrally-scaled CMDS output satisfies $$\|\hat{U}\hat{S}^{1/2}-US^{1/2}R\|_F\leq B\lambda_1^{1/2}(E_{\varphi})\left(2+4B\kappa^{1/2}+(1+2B)\frac{m\log(n)}{\sqrt{n}\lambda_c(E_{\varphi})}\right),$$
where $\kappa=\lambda_1(E_{\varphi})/\lambda_c(E_{\varphi})$. In particular, we have
$$\sum_{i=1}^{m} \|\hat{\psi}(t_i)-R\psi(t_i)\|^2\leq B^2\lambda_1(E_{\varphi})\left(2+4B\kappa^{1/2}+(1+2B)\frac{m\log(n)}{\sqrt{n}\lambda_c(E_{\varphi})}\right)^2.$$
\end{thm}

 If all but the top $c$ eigenvalues of $\mathcal{D}_{\varphi}$ are sufficiently small---as is the case when $\mathcal{D}_{\varphi}$ is rank $c$--- Theorem \ref{thm:cmds} ensures that a Euclidean mirror can be consistently estimated. As such, if the important aspects of a finitely-sampled latent position process, such as changepoints or anomalies, are reflected in low-dimensional Euclidean space, then we recover an optimal finite-sample mirror consistently through CMDS applied to the estimated distance matrix. We encapsulate our consistency results and connections between true distances, their estimates, and associated Euclidean mirrors in Figure \ref{fig:commuting}.

\begin{figure}[h]
    \centering
    \includegraphics[width=\textwidth]{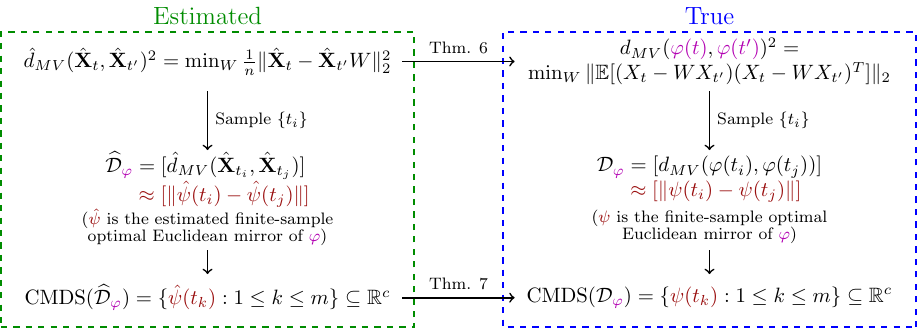}
    \caption{Consistent estimation of network dissimilarity and Euclidean mirrors.}
    \label{fig:commuting}
\end{figure}

The right-hand side of Figure \ref{fig:commuting} lists the true and typically unobserved distance measure $d_{MV}$, and from it, immediately below, the matrix of pairwise distances $\mathcal{D}_{\varphi}$. If this dissimilarity is is Euclidean realizable in $c$ dimensions, then classical multidimensional scaling will recover this mirror, denoted by $\psi$, up to Euclidean distance-preserving transformations.

On the left-hand side of Figure \ref{fig:commuting}, we see how to compute an estimate of $d_{MV}$ from spectral embeddings of a pair of observed network adjacencies. Theorem \ref{thm:approximation} grants that the estimated dissimilarity matrix of pairwise distance $\hat{\mathcal{D}}_{\varphi}$ will be close to the true dissimilarity $\mathcal{D}_{\varphi}$, and if the latent position process is Euclidean realizable, Theorem \ref{thm:cmds} establishes that classical multidimensional scaling applied to $\hat{\mathcal{D}}_{\varphi}$ serves as a consistent estimate for $\psi$.

%Grand unified theory
This figure describes our overall approach to the problem of inference in time series of networks: first, we construct a useful dissimilarity measure that captures important features of the underlying LPP; next, show how this dissimilarity can be consistently estimated; and finally, extract an estimated mirror that provides a low-dimensional representation of network evolution. Our methodology is not restricted simply to the specific distance measure $d_{MV}$ that we have defined here, and other notions of distance may have conceptual, theoretical, or computational benefits depending on the underlying process.

\begin{remark}
Using CMDS on the distance matrix as our estimate means that permuting the time points $t_1, \cdots, t_k$ to $t_{\sigma(1)}, \cdots, t_{\sigma(k)}$ where $\sigma \in S_k$ produces the same set of points in $\RR^c$ (up to an orthogonal transformation), in a permuted order. As such, \emph{if the original times are associated with these points}, we recover the identical smooth curve. If the times are not retained, only the original time ordering recovers this smooth trajectory.

This does not change the content of the theorem: one just replaces $t_i$ with $t_{\sigma(i)}$ in the statement. On the other hand, if we have an {\em unordered} collection of networks, and not a time series of graphs ordered naturally by time, our estimation procedure will yield a collection of points in $\RR^c$, but these will not typically fall on a 1-dimensional curve. \hfill 
$\square$
\end{remark}

When we have exact Euclidean realizability or when the tail eigenvalues of $\mathcal{D}_{\varphi}$ can be bounded directly, we obtain the following two corollaries of Theorem \ref{thm:cmds}.

%specialize Theorem~\ref{thm:cmds} to two cases where the sum of the tail eigenvalues can be bounded directly from realizability assumptions on $\varphi$. The first considers the case of exact Euclidean realizability.

\begin{corollary}
\label{cor:edm}
Suppose $\mathcal{D}_{\varphi}$ is a Euclidean distance matrix with dimension $c$. Then $\mathcal{D}_{\varphi}=\mathcal{D}_{\psi}$, and for fixed $m$, there is a constant $C=C(m)$ such that with overwhelming probability, there is a real orthogonal matrix $R\in\mathcal{O}^{c\times c}$ such that the CMDS output satisfies $$\|\hat{U}\hat{S}^{1/2}-US^{1/2}R\|_F\leq C\frac{\log(n)}{\sqrt{n}}.$$ In particular, we have
$$\sum_{i=1}^{m} \|\hat{\psi}(t_i)-R\psi(t_i)\|^2\leq C\frac{\log^2(n)}{n}.$$
\end{corollary}

Suppose that $\psi$ is Lipschitz continuous with constant $L$ and $\varphi$ has realizability constant $B$. If we further assume that there exists a constant $A$ such that $d_{MV}(\varphi(t),\varphi(s))\leq A$ and  $\|\psi(t)-\psi(s)\|_2\leq A$ for all $s,t\in[0,T]$, then we can bound the sum of the tail eigenvalues of $\mathcal{D}_{\varphi}^{(2)}$, turning the approximate Lipschitz Euclidean realizability assumption into an eigenvalue bound. 
%Note that if $\psi$ is Lipschitz continuous with constant $L$ and $\varphi$ has realizability constant $B$, 
Note that we can can always choose $A=(2L+B)T$ from the realizability assumptions, but in certain cases, $A$ may be smaller, and in particular, not grow linearly with T. While this corollary is stated for the Lipschitz case, a version of it may be formulated for the $\alpha$-H\"{o}lder case as well.

\begin{corollary}
\label{cor:uniform}
Suppose $\mathcal{D}_{\varphi}$ is approximately Lipschitz Euclidean $c$-realizable with realizability constant $B$. Suppose $d_{MV}(\varphi(t),\varphi(s)), \|\psi(t)-\psi(s)\|\leq A$ for all $s,t\in[0,T]$. Suppose that $t_i=Ti/m$ for $i=1,\ldots,m$. Then 
$$\|\mathcal{D}_{\varphi}^{(2)}-\mathcal{D}_{\psi}^{(2)}\|_F\leq 2ABT\sqrt{(m^2-1)/6}\leq 0.82 ABT m.$$

For $m$ fixed, there is a constant $C=C(m)$ such that with high probability, there is a rotation matrix $R\in\mathcal{O}^{c\times c}$ such that the CMDS output satisfies
$$\|\hat{U}\hat{S}^{1/2}-US^{1/2}R\|_F \leq C\left(\frac{\log(n)}{\sqrt{n}}+0.82ABT\right).$$

\end{corollary}

\begin{example}
\label{ex:lipbnds}
In Example~\ref{ex:integratedBM}, which is approximately Lipschitz Euclidean realizable, we have $$B\leq \frac{\sigma^2T}{a\|v\|+\sqrt{a\|v\|+\sigma^2T}},$$ and $A\leq T\sqrt{a^2\|v\|^2+\sigma^2T}.$ 
\end{example}

\begin{remark}
    The relationship between the true and estimated network features from Figure~\ref{fig:commuting} is equally appropriate for certain changes to the distance metric. For example, consider a latent position process with $\varphi(t)=c(t)\varphi(0)$, $c(t)\in[0,1]$, corresponding to a global change in the density of the network, but one that leaves the community structure unchanged. 
    
    Using the adjacency spectral embedding with the eigenvectors scaled by the eigenvalues will detect these global transformations in sparsity, while using unit eigenvectors of the adjacency spectral embedding (unscaled by their respective eigenvalues) will ignore changes of this type, instead focusing only on divergences in community structure. These different computations of network dissimilarity will result in distinct mirrors, highlighting distinct changes in the networks over time. \hfill $\square$
\end{remark}
Together, these theorems ensure that time-dependent underlying low-dimensional structure associated to network evolution can be consistently recovered. In what follows, we will see how this methodology can be employed in real and synthetic data to reveal important structural features and potential anomalies in network time series.

\section{Experiments}
\label{sec:experiments}
\subsection{Organizational network data and pandemic-induced shifts}\label{sec:org_sci_data}
We start with a discussion of the visualization of the communication network in Figure~\ref{fig:org_sci_example} and the output of our mirror-based analysis in Figure~\ref{fig:org_sci_us}. We consider a time series of weighted communication networks, arising from the email communications between 32277 entities in a large organization, with one network generated each month from January 2019 to December 2020, a period of 24 months. This data was studied through the lens of modularity in \cite{zuzul2021dynamic}. 

To generate the visualization in Figure~\ref{fig:org_sci_example}, we first cluster the graph $G_1$ using Leiden clustering and then apply Node2Vec \citep{grover2016node2vec} to each graph $G_t$ to obtain embeddings $\hat{\bX}_t \in \RR^{n\times 33}$. Applying the Uniform Manifold Approximation and Projection (UMAP) algorithm of \cite{mcinnes2018umap} to the $\hat{\bX}_t$ matrices, we obtain a layout with points $\hat{\bX}_t'\in \RR^{n\times 2}$, colored using the labels from the Leiden clustering.

To obtain our mirror estimates, we again apply Leiden clustering \citep{traag2019louvain} to the January 2019 network, obtaining 33 clusters that we retain throughout the two year period. We make use of this clustering to compute the Graph Encoder Embedding (GEE) of \cite{shen_GEE} for each time, which produces spectrally-derived estimates of invertible transformations of the original latent positions. For each time $t=1, \cdots, 24$, we obtain a matrix $\hat{\bZ}_t\in\RR^{32277\times 33}$, each row of which provides an estimate of these transformed latent positions. Constructing the distance matrix $\widehat{\mathcal{D}}_{\varphi}=[\hat{d}_{MV}(\hat{\bZ}_t,\hat{\bZ}_{t'})]\in\RR^{24\times24}$, we apply CMDS to obtain the estimated curve $\hat{\psi}$ shown in the left panel of Figure~\ref{fig:org_sci_us}, where the choice of dimension $c=2$ is based on the scree plot of $\widehat{\mathcal{D}}_{\varphi}$. The nonlinear dimensionality reduction technique ISOMAP \citep{isomap_science}, which relies on a spectral decomposition of geodesic distances, can be applied to these points to extract an estimated 1-dimensional curve, which we plot against time in the right panel of Figure~\ref{fig:org_sci_us}. Since the ISOMAP embedding generates points whose Euclidean distances approximate the geodesic distances between points on the mirror, larger changes in the $y$-axis of this figure correspond to significant changes in the networks. This one-dimensional curve exhibits some changes from the previous trend in Spring 2020 and a much sharper qualitative transformation in July 2020. What is striking is that both these qualitative shifts correspond to policy changes: in Spring 2020, there was an initial shift in operations, widely regarded at the time as temporary. In mid-summer 2020, nearly the peak of the second wave of of COVID-19, it was much clearer that these organizational shifts were likely permanent, or at least significantly longer-lived.
\begin{figure}
    \centering
    \includegraphics[width=0.40\textwidth]{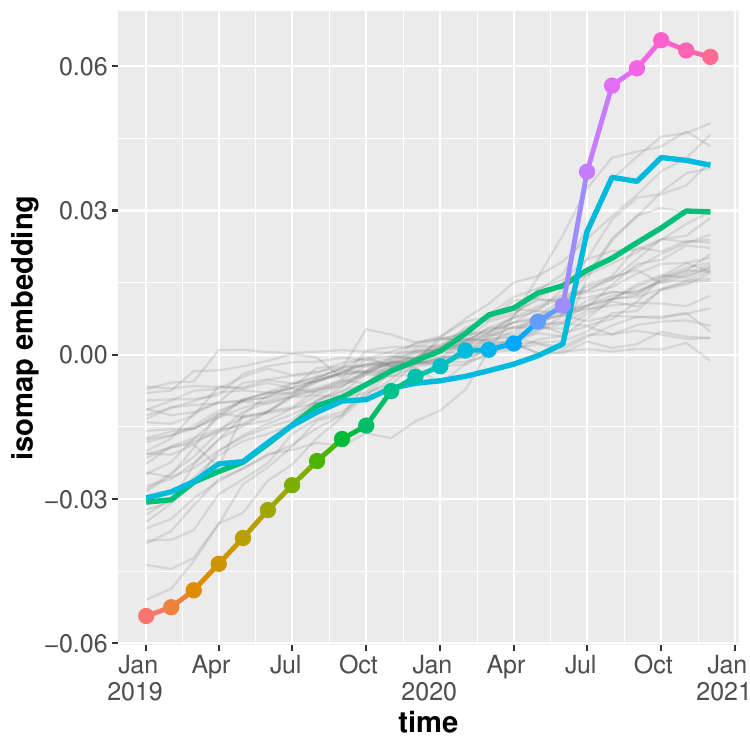}
\subfloat{\includegraphics[width=0.40\textwidth]{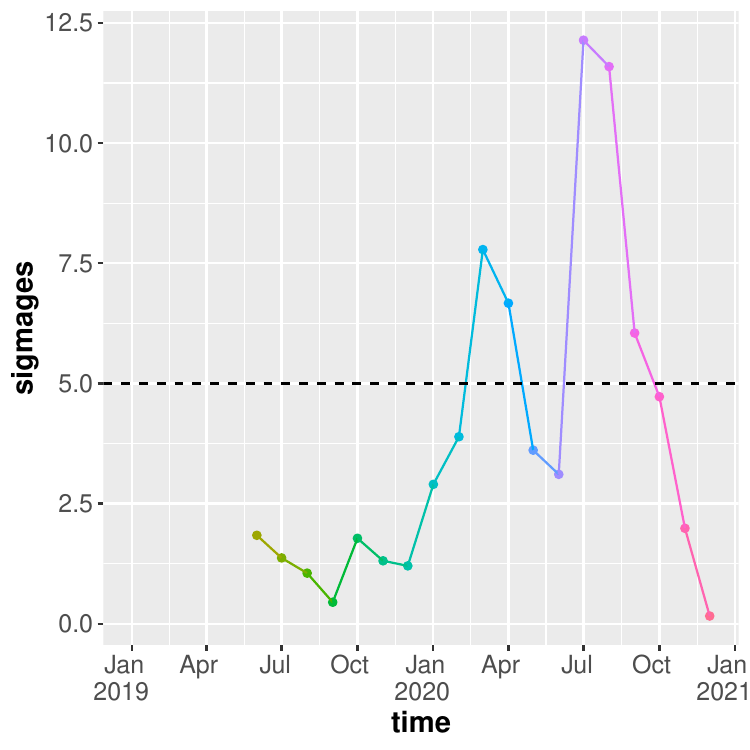}
\includegraphics[width=0.40\textwidth]{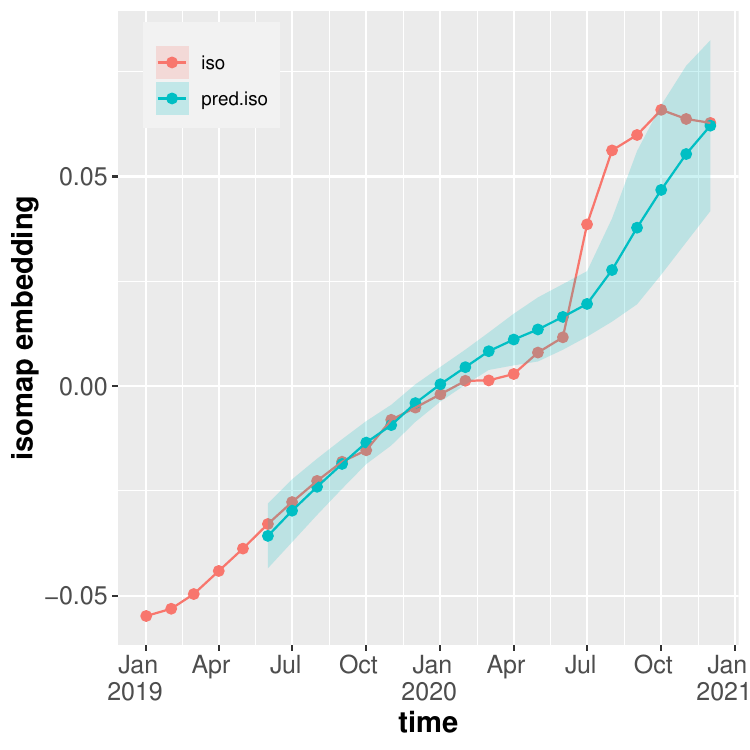}}
    \caption{
Differential pandemic effect on 33 subcommunities of the larger organization and change point detection.
Top: For each subcommunity, we apply our methods to obtain an ISOMAP representation (single grey curve).
Green and blue curves correspond to communities highlighted in Figure~\ref{fig:org_sci_example}
which exhibit different changes over the time period.
Overall network ISOMAP embedding from Figure~\ref{fig:org_sci_us} is overlaid for context.
Bottom left panel: Comparing ISOMAP embedding at a given time to mean of previous 5 months' embeddings, measured in terms of their standard deviation. Plot of sigmages indicates clear outliers in March-April of 2020, and again in July-September. Bottom right panel: Observed ISOMAP embedding along with a running confidence interval for the predicted value of the ISOMAP embedding generated from linear regression applied to the embeddings for the previous 5 months, with width 5 standard deviations.}
    \label{fig:subcommunities}
\end{figure}
In Figure~\ref{fig:subcommunities}, the top panel plots the result of our methods applied to the induced subgraphs corresponding to each of the 33 communities; these are represented by the grey trajectories. The trajectories of two subcommunities have been highlighted in the top plot: the green curve shows a constant rate of change throughout the two-year period, and does not exhibit a noticeable pandemic effect. The blue curve, on the other hand, shows a significant flattening in early 2020, followed by rapid changes in summer. Thus, we can see a differential effect of the pandemic on different work groups within the organization. In Figure~\ref{fig:subcommunities}, both bottom panels show methods for identifying changepoints over the 24 months, with consistent results. We start by generating the ISOMAP embedding $\iota$ of $\hat{\psi}$, yielding $\iota_t\in \RR$ for $t=1,\ldots,24$. %In the left panel, we plot the fitted values of ISOMAP applied to the CMDS output of $\mathcal{D}_{\hat{\psi}}$ against the residuals. The blue curve indicates a line of best fit, and several outliers become apparent in the spring/summer months of 2020. In the middle panel, 
In the bottom left panel, for each time starting in June 2019, we plot the \emph{sigmage} (see \cite{good1992bayes}) of its ISOMAP embedding relative to the previous 5 months. That is, we measure the distance of the ISOMAP embedding to the mean of the previous 5 months' embeddings, relative to the standard deviation of those embeddings: letting $\mu_t=\frac{1}{5}\sum_{i=1}^5 \iota_{t-i}$ for each time $t>5$, the sigmage $s_t$ is given by
$$ s_t = \frac{|\iota_t - \mu_t|}{\sqrt{\frac{1}{4}\sum_{i=1}^5(\iota_t-\mu_t)^2}}. $$
Note that since the computation of the sigmages require a window of time-points, we are only able to produce these estimates starting in June 2019. We see apparent outliers in March and April, and again in July-September 2020. The right panel shows the ISOMAP curve with a moving prediction confidence interval of width 5 standard deviations, generated from simple linear regression applied to the previous 5 time points (which is why we again only have an interval starting in June). This method indicates the same set of outliers as the previous one, but allows for some more detailed analysis: In March and April 2020, it appears that the behavior is anomalous because the network \emph{stopped} drifting, while the behavior in the summer of that year is anomalous because it made a significant jump from its previous position.

%Several follow-on questions about the data suggest themselves: Which sub-communities are responsible for the change in behavior, and which ones change the most during this time? Can we detect the changepoints consistently and develop hypothesis testing procedures for determining whether a qualitative shift corresponds to a changepoint?

In Section \ref{sec:realdataviz}, we consider additional visualizations of the organizational communication networks. In Figure~\ref{fig:graphstats} of Section \ref{sec:realdataviz}, we plot a collection of other summary statistics, namely edge counts, maximum degree, median degree, and modularity, for each network over time. As we describe in that section, since this approach considers each network separately, these summary statistics exhibit greater variance than the ISOMAP embedding of the mirror (Figure~\ref{fig:org_sci_us} right panel, or Figure~\ref{fig:subcommunities} bottom right panel), and they do not capture changepoints associated to pandemic policy restrictions. 

% In addition, seasonal effects play a greater role in these plots, which add to the difficulty in detecting the changepoints. Note that in contrast to the mirror visualizations in Figures \ref{fig:subcommunities} and \ref{fig:changepoint}, none of the plots in Fig. \ref{fig:graphstats} allows for easy qualitative visualization of two important changepoints driven by company policy at the start of the pandemic restrictions (Spring 2020) and the change in the imposition of restrictions from short-term to open-ended and longer-term (July 2020).

% \begin{figure}[h]
%     \centering
%     \includegraphics[width=0.7\textwidth]{Figures/gstat-jasa.pdf}
%     \caption{Additional graph summary statistics over time. These methods consider each graph separately, rather than our mirror approach which accounts for dependence across time: as a result, the curves show greater variance and seasonal effects, obfuscating the changepoints that are captured by the mirror in Figures \ref{fig:subcommunities} and \ref{fig:changepoint}.}
%     \label{fig:graphstats}
% \end{figure}

% In Appendix~\ref{sec:realdataviz}, we include additional visualizations of the networks for January, May, and September of 2019 and 2020. 

\begin{figure}[h]
    \centering
 \includegraphics[width=0.4\textwidth]{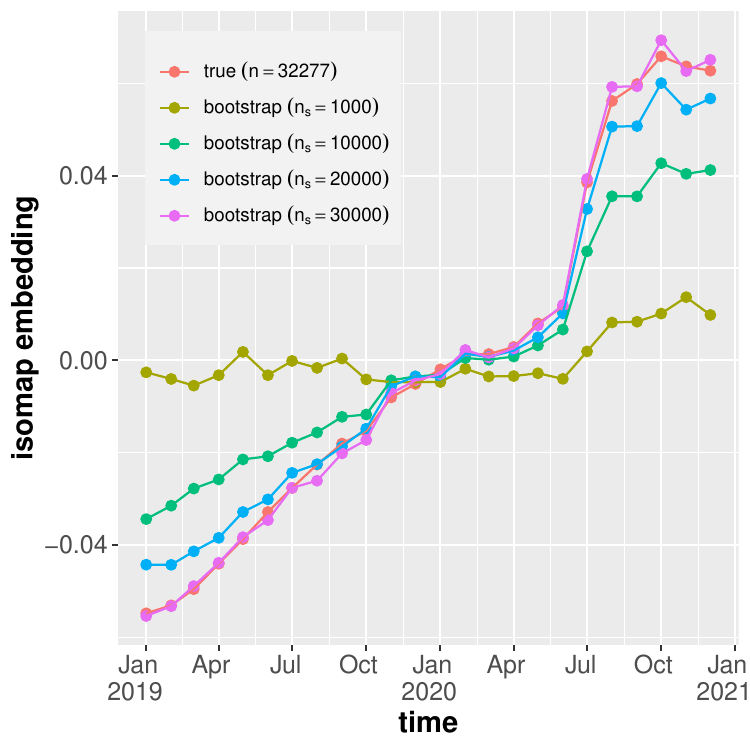}
    \caption{Convergence of bootstrapped estimates. We collect four bootstrapped time series of graphs at various network sizes $n_s$, sampling the latent positions with replacement from the real data of the previous section, and generating the networks as random dot product graphs with these synthetic latent position matrices. ISOMAP on the resulting network embeddings shows that the bootstrapped curve converges to the original ISOMAP curve as the number of samples $n_s$ increases.}
    \label{fig:bootstrap}
\end{figure}

\subsection{Synthetic data and bootstrapping}
\label{sec:bootstrap}

In the previous section, we apply the GEE embedding to obtain estimates $\hat{\bZ}_t$, which are then used for estimates of pairwise distances  $\hat{d}_{MV}(\hat{\bZ}_t,\hat{\bZ}_s)$. Although the GEE differs slightly from the adjacency spectral embedding, it is computationally more tractable and yields similarly useful output. To further illustrate our underlying theory, however, we consider \emph{synthetic data}. That is, we use real data to obtain a distribution from which we may resample. Such a network bootstrap permits us to test our asymptotic results through replicable simulations that are grounded in actual data. To this end, we consider the true latent position distribution at each time to be equally likely to be any row of the GEE-obtained estimates from the real data, $\hat{\bZ}_t\in\RR^{32277\times 33}$, for $t=1,\ldots,24.$ Given a sample size $n_s$, for each time, we sample these rows uniformly and with replacement to get a matrix of latent positions $\bX_t\in\RR^{n_s\times 33}$. We treat this matrix as the generating latent position matrix for independent adjacency matrices  $A_t\sim\mathrm{RDPG}(\bX_t)$. Note that if for sample $i$, we choose row $j$ of $\hat{\bZ}_1$ at time $t=1$, then the same row $j$ of $\hat{\bZ}_t$ will be used for all times $t=1,\ldots,24$ for that sample, so that the original dependence structure is preserved across time. We may now apply the methods described in our theorems, namely ASE of the adjacency matrices followed by Procrustes alignment, to obtain the estimates $\hat{\bX}_t$, along with the associated distance estimates. In Figure~\ref{fig:bootstrap}, we see that ISOMAP applied to the CMDS embedding of the bootstrapped data converges to the original ISOMAP curve, as predicted by our theorems.

\begin{figure}[h]
    \centering
    \includegraphics[width=0.4\textwidth]{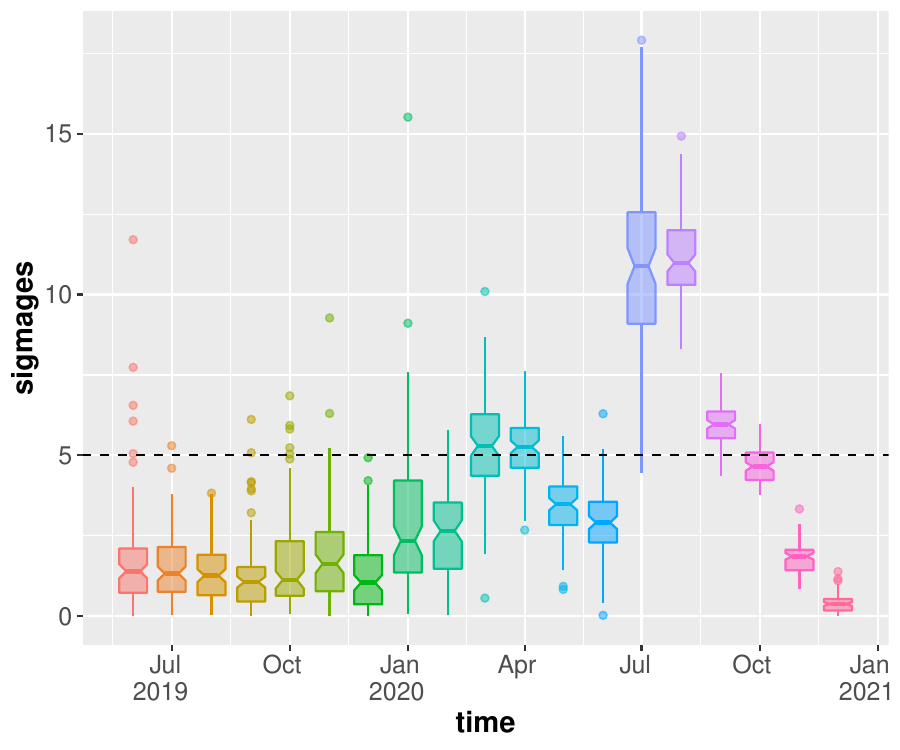}
    \caption{Pandemic effect recovered from synthetic data. For each of 100 replicates of bootstrapped data, with $n_s=30000$ for each replicate, we repeat procedure in bottom left panel of Figure~\ref{fig:subcommunities}. Sigmages plotted in a box-and-whisker plot. Pandemic effect in summer of 2020 is visible in all but a few replicates; effect in March-April is still identified in the majority of replicates.}
\label{fig:syntheticchangepoint}
\end{figure}

To check whether this procedure demonstrates the pandemic effects, in Figure~\ref{fig:syntheticchangepoint}, we show the sigmages for each month, plotted over 100 replicates of this experiment, with $n_s=30000$ for each replicate. The pandemic effect in summer of 2020 is clearly visible in all but a few replicates, while the effect in March-April is still identified in the majority of replicates. We observe dramatic changes in variance for certain months, over the different replicates: this might indicate the discrepancies between the pandemic effect on different network entities, rendering the final estimate much more sensitive to the sample of rows used to generate the network.

In Section \ref{sec:sbms}, we provide mirror estimates for an evolving stochastic blockmodel with a change in rank from 2 to 1. Figures \ref{fig:sbms-distances} and \ref{fig:sbms-mds1and2} demonstrate the mirror estimation procedure and model misspecification in the embedding dimension, specifically the accuracy of the first dimension of the estimated distance matrix and the noise in the second dimension at and after the collapse to a rank 1 model.

Code for these results and additional simulations can be found at\\
\mbox{\url{https://www.cis.jhu.edu/~parky/IsoMirror/dynamics.html}.}

\section{Discussion}
\label{sec:discussion}
To effectively model time series of networks, it is natural to consider network evolution governed by underlying low-dimensional dynamics. Here, we examine latent position networks in which the vertex latent positions follow a stochastic process known as a latent position process (LPP). Under mild conditions, we can associate to the LPP certain geometric structure, and understanding how that structure changes with time allows us to identify transformations in network behavior across multiple scales. To make this precise, we define the maximum directional variation norm and metric on the space of random latent positions. We describe notions of Euclidean realizability and Euclidean mirrors for this metric and process, characterizing how closely this metric can be approximated by a Euclidean distance. 
% If the latent position process is sufficiently regular, its range is a one-dimensional manifold, and we detail several useful examples of such processes. 
Of course, the latent position process is typically unobserved; what we have instead is a time series of networks from which these latent positions must be estimated. One of our key results is that the pairwise dissimilarity matrix of maximum directional variation distances between latent positions $X_t$ and $X_{t'}$ at pairs of time points can be consistently estimated by spectrally embedding the network adjacencies at these different pairs of times and computing spectral norm distances between these embeddings. When the latent position process is such that the maximum directional variation metric between any pair of latent positions $X_t, X_{t'}$ is approximately Euclidean realizable, we find that classical multidimensional scaling applied to the estimated distances gives us an inferentially valuable low-dimensional representation of network dissimilarities across time. Further dimension reduction techniques, such as ISOMAP, can further clarify changes in network dynamics. To this last point, ISOMAP is a manifold-learning algorithm; detailed analysis of its effect on embeddings of estimated pairwise network distances can bring us closer to provable guarantees for change-point detection. More broadly, the interplay between the probabilistic structure of the underlying latent position process and the geometric structure of the Euclidean mirror is a key component of the estimated Euclidean representation of relationships between networks across time. 
% further investigation of ISOMAP or other manifold-learning techniques to the output of CMDS in these settings may yield significant insights that enable hypothesis testing for changepoints and anomaly detection. The interplay between the probabilistic structure of the underlying latent position process and the geometric structure of the Euclidean mirror is another rich area for future study. 

We consider two estimates for the maximum directional variation distance $d_{MV}$, namely the spectral norm applied to the GEE estimate, and the $\hat{d}_{MV}$ estimated distance between the adjacency spectral embeddings. However, these are far from the only options, and it is an open question whether $d_{MV}$ or another metric on the space of random variables is best for downstream inference tasks under certain model assumptions. It is also an open question whether there is a better estimate for the distance $d_{MV}$ itself, either in terms of computational complexity or statistical properties. Of particular interest is the spectral norm distance applied to the omnibus embeddings \citep{levin_omni_2017} for the adjacency matrices: this likely converges to another distance on the space of random variables, potentially highlighting different features in the final CMDS embedding. Results quantifying the distribution of the errors in the CMDS embedding are key to formulating hypothesis tests for changepoint detection. The perspective described in Figure~\ref{fig:commuting}, which connects distance metrics for generative processes of networks to their estimates, translating manifold geometry into Euclidean geometry, is a useful contribution to time series analysis for networks. It provides mathematical formalism for network dynamics; asymptotic properties of estimates of manifold structure; and conditions for the representation of time-varying networks in low-dimensional space. Latent position networks are interpretable, estimable, and flexible enough to capture important features of real-world network time series. As such, this canonical framework invites and accommodates future approaches to joint network inference. 

%latent position processes underlying the networks. In light of these features, we believe this is the ideal framework for future study of network time series.

% \section{Acknowledgements}
 %All authors gratefully acknowledge funding from Microsoft Research, the Naval Engineering Education Consortium, the United States National Science Foundation (SES-1951005), and the Acheson Duncan Fund.

\bibliographystyle{chicago}
\bibliography{TSG_dynamics.bib}

\appendix
\newpage
\section{Supplementary Material: Proofs, supporting results, and additional simulations}\label{sec:proofs}  
\subsection{Proofs and supporting results for Section~\ref{sec:model}}
\label{sec:proofs2}

\begin{lemma}
\label{lemma:metric}
The function $d_{MV}(X,Y)$ is a metric on the space of random variables, up to the equivalence relation where $X\sim Y$ if there is some $W\in\mathcal{O}^{d\times d}$ such that $X=WY$ almost surely.
\end{lemma}
\begin{proof}
Recall that $d_{MV}(X,Y)$ is defined as
$$d_{MV}(X,Y)=\min_{W}\|\EE[(X-WY)(X-WY)^\top]\|_2^{1/2}=\min_W\|\EE[(W^\top X-Y)(W^\top X-Y)^\top]\|_2^{1/2}.$$
Clearly, this is symmetric and nonnegative. The triangle inequality holds, since for any $Q\in\mathcal{O}^{d\times d}$ and $Z\in L^2(\Omega)$, we have
\begin{align*}
d_{MV}(X,Y)^2&=\min_W \max_u u^\top\EE[(X-QZ+QZ-WY)(X-QZ+QZ-WY)^\top]u\\
&=\min_W \max_u\EE[\langle u,X-QZ\rangle^2+2\langle u,X-QZ\rangle\langle u,QZ-WY\rangle+\langle u,QZ-WY\rangle^2]\\
&\leq \min_W \max_u (\EE[\langle u,X-QZ\rangle^2]^{1/2}+\EE[\langle u,QZ-WY\rangle^2]^{1/2})^{2},
\end{align*}
by the Cauchy-Schwartz inequality applied to the $L^2(\Omega)$-inner product. This is further bounded as
\begin{align*}
d_{MV}(X,Y)&\leq \min_W \left(\max_u \EE[\langle u,X-QZ\rangle^2]^{1/2}\right)+\left(\max_v \EE[\langle v,QZ-WY\rangle^2]^{1/2}\right)\\
&=\left(\max_u \EE[\langle u,X-QZ\rangle^2]^{1/2}\right)+\min_W\left(\max_v \EE[\langle v,Z-(Q^\top W)Y\rangle^2]^{1/2}\right).
\end{align*}
Since $\mathcal{O}^{d\times d}=Q^\top \mathcal{O}^{d\times d}$ when $Q\in \mathcal{O}^{d\times d}$, the latter term is just $d_{MV}(Z,Y)$. Since this upper bound holds for any $Q\in\mathcal{O}^{d\times d}$, it must also hold for the minimizer.

Now suppose that $d_{MV}(X,Y)=0$. Since the spectral norm is a norm, this tells us that for the $W\in\mathcal{O}^{d\times d}$ achieving the minimum, $\EE[(X-WY)(X-WY)^\top]=0$. Since $(X-WY)(X-WY)^\top$ is positive semidefinite for every $\omega\in\Omega$, this implies that $X=WY$ almost surely.
\end{proof}

{\bf Proof of Theorem~\ref{thm:optimal} Existence}. Given that $\varphi$ is approximately Euclidean $c$-realizable, there are some $\alpha,C>0$ and $\alpha$-H\"{o}lder mirror $\psi$ satisfying the approximate realizability bound. Consider the set of mirrors $\mathcal{S}(c,\alpha,C)$ whose elements are functions $\psi\in\mathcal{C}^{\alpha}([0,T],\RR^c)$ satisfying two requirements: first, the realizability bound with constant $C$, and second, $\int_0^T \psi(t)\,\mathrm{d}t=0$. 

This class is equicontinuous and uniformly bounded. Indeed, because $\varphi$ is approximately $\alpha$-H\"{o}lder $c$-realizable, we derive that 
$d_{MV}(\varphi(t),\varphi(s))\leq L|t-s|^{\alpha}$, and the realizability bound implies that $\|\psi(t)-\psi(s)\|\leq (C+L)|t-s|^{\alpha}$ for all $s,t\in[0,T]$, which guarantees equicontinuity. Since $\int_0^T \psi(t) \mathrm{d}t=0$ and that $\psi$ all have a common $\alpha$-H\"older constant $C+L$, we derive that for each $t \in [0, T]$, there exists $M_t$ with
$$\|\psi(t)\| \leq M_t\leq (L+C)T^\alpha$$
for all $\psi$, guaranteeing uniform boundedness.

Consider any sequence of $\mathbb{R}^c$-valued functions $\psi_n\in\mathcal{S}(c,\alpha,C)$ converging uniformly to $\psi\in\mathcal{C}([0,T],\RR^n)$. For any $s,t\in[0,T]$, uniform convergence guarantees that 
\begin{align*}
\|\psi(t)-\psi(s)\|&=\lim_{n\rightarrow\infty}\|\psi_n(t)-\psi_n(s)\|\leq (C+L)|t-s|^{\alpha}\\
|d_{MV}(\varphi(t),\varphi(s))-\|\psi(t)-\psi(s)\||&=\lim_{n\rightarrow\infty}|d_{MV}(\varphi(t),\varphi(s))-\|\psi_n(t)-\psi_n(s)\||\leq C|t-s|^{\alpha},
\end{align*}
establishing that $\psi\in \mathcal{S}(c,\alpha,C)$ as well. Thus $\mathcal{S}(c,\alpha,C)$ is a closed subset of $\mathcal{C}([0,T],\RR^c)$, and by the Arzel\`a-Ascoli Theorem, it is compact. Since $\mathcal{L}(\psi)$ is a continuous functional on $C[0,T]$ and $\mathcal{S}(c, \alpha, C)$ is compact, there exists some $\psi \in \mathcal{S}$ at which the minimum of $\mathcal{L}$ over $\mathcal{S}$ is achieved. \hfill $\square$

To prove Theorem~\ref{thm:unique}, we introduce appropriate definitions and supporting lemmas. Let the kernel $\kappa:[0,T]^2\rightarrow\RR$ be defined as
\begin{align*}
\kappa(s,t)&= -\frac{1}{2}\left[d_{MV}(\varphi(s),\varphi(t)) - \frac{1}{T}\int_0^T d_{MV}(\varphi(u),\varphi(t))\,\mathrm{d}u\right.\\
&\qquad\quad \left.-\frac{1}{T}\int_0^T d_{MV}(\varphi(s),\varphi(v))\,\mathrm{d}v+\frac{1}{T^2}\int_0^T\int_0^T d_{MV}(\varphi(u),\varphi(v))\,\mathrm{d}u\,\mathrm{d}v \right].
\end{align*}
We define the integral operator $I:L^2([0,T])\rightarrow L^2([0,T])$ as
$$
I[f](x) = \frac{1}{T}\int_0^T \kappa(x,y) f(y)\,\mathrm{d}y.
$$
We take the normalized inner product on $L^2([0,T])$, namely $\langle f,g\rangle_{L^2} = \frac{1}{T}\int_0^T f(t)g(t)\,\mathrm{d}t$. It is easy to show that this operator is self-adjoint and Hilbert-Schmidt.

\begin{lemma}
\label{lem:holder}
For any $f\in L^2([0,T])$, $I[f]$ is $\alpha$-H\"older continuous. Letting $L$ denote the $\alpha$-H\"older constant for $\varphi$, any normalized eigenfunction $f$ satisfying $I[f]=\lambda f$, $\lambda>0$ is $\alpha$-H\"older with constant at most $L/\lambda$.
\end{lemma}
\begin{proof}
We bound the values of $I[f]$ as
\begin{align*}
|I[f](s)-I[f](t)|&=\left| \frac{1}{T}\int_0^T (\kappa(s,y)-\kappa(t,y))f(y)\,\mathrm{d}y\right|\\
&\leq \frac{1}{T}\int_0^T |\kappa(s,y)-\kappa(t,y)||f(y)|\,\mathrm{d}y\\
&\leq \|\kappa(s,\cdot)-\kappa(t,\cdot)\|_{L^2}\|f\|_{L^2}\\
&\leq L\|f\|_{L^2}|t-s|^{\alpha}.
\end{align*}
Here we use the H\"older property of $\varphi$ to see that $|\kappa(s,y)-\kappa(t,y)|\leq L|t-s|^{\alpha}$ for all $y\in[0,T]$, giving $\|\kappa(s,\cdot)-\kappa(t,\cdot)\|_{L^2}\leq L|t-s|^{\alpha}$. When $f$ is a normalized eigenfunction for $I$ with eigenvalue $\lambda$, this gives us 
\begin{align*}
|f(s)-f(t)|&= \frac{1}{\lambda}|I[f](s)-I[f](t)|\\
&\leq \frac{L\|f\|_{L^2}}{\lambda}|t-s|^{\alpha},
\end{align*}
which gives the desired bound since $\|f\|_{L^2}=1$. 
\end{proof}

The following lemma completes the proof of Theorem~\ref{thm:unique}.

\begin{lemma}
\label{lem:eigenexpansion}
Suppose $\varphi$ is exactly $\alpha$-H\"older Euclidean $c$-realizable. Then $I$ is finite rank, and may be written as $$I[f] = \sum_{i=1}^c \lambda_i \langle g_i, f\rangle_{L^2([0,T])} g_i,$$ where $\lambda_i\geq 0, 1\leq i\leq c$. Moreover, if $\psi(t)$ is defined by
$$\psi(t) = \sum_{i=1}^c \sqrt{\lambda_i}g_i(t) {\bf e}_i$$
where ${\bf e}_i$ is the $i$th basis vector in $\mathbb{R}^c$, then $\psi$ is the unique minimizer to the variational problem \ref{eq:var_prob_L} (up to orthogonal transformation), and satisfies $\mathcal{L}(\psi)=0$. 
\end{lemma}
\begin{proof}
By definition of exact realizability, there is an $\alpha$-H\"older function $\psi:[0,T]\rightarrow\RR^c$ such that $d_{MV}(\varphi(s),\varphi(t))=\|\psi(s)-\psi(t)\|_2$ for all $s,t\in[0,T]$, which ensures that the minimum value of $\mathcal{L}(\psi)$ over $\mathcal{S}(c,\alpha,0)$ must be 0. Since the function $\psi$ defined above satisfies $d_{MV}(\varphi(s),\varphi(t))=\|\psi(s)-\psi(t)\|$ for all $s,t\in[0,T]$, it suffices to show that any other such function satisfying this condition is just an orthogonal transformation of $\psi$. Consider an enumeration of the rationals in $[0,T]$, $r_1,r_2,\ldots$. Since $$\|\psi_1(s)-\psi_1(t)\|=d_{MV}(\varphi(s),\varphi(t))=\|\psi_2(s)-\psi_2(t)\|,$$ for any $k\geq 1$, we get that the squared distance matrices $\mathcal{D}_i = [\|\psi_i(r_j)-\psi_i(r_\ell)\|^2]_{j,\ell=1}^k$ are equal, and thus the Gram matrices $-\frac{1}{2}P\mathcal{D}_iP^T$ are, also. There exist orthogonal matrices $V_k, W_k\in M_c(\RR)$ giving 
$$[\psi_1^T(r_j)]_{j=1}^k = L_k V_k,\quad [\psi_2^T(r_j)]_{j=1}^k = L_k W_k,$$ where $L_k\in M_{k,c}(\RR)$ is the Cholesky factor of $-\frac{1}{2}P\mathcal{D}_iP^T$ for the set of times $\{r_1,\ldots,r_k\}$. This implies the existence of a unitary $R_k\in M_c(\RR)$ such that $\psi_1(r_j)=R_k \psi_2(r_j)$ for all $1\leq j\leq k$. Since the set of unitary matrices is compact, there is a subsequence of the matrices $R_k$ that converges to some unitary matrix $R$. But since $\psi_1(r_j)=R_k\psi_2(r_j)$ for all $k\geq j$, we necessarily have $\psi_1(r_j)=R\psi_2(r_j)$ for every $j\geq 1$. Since the $\{r_j\}$ form a dense subset of $[0,T]$ and $\psi_i$ is continuous for $i=1,2$, we get $\psi_1(s)=R\psi_2(s)$ for all $s\in[0,T]$. 
\end{proof}

Suppose $\varphi:[0,T]\rightarrow(L^2(\Omega),d_{MV})$ %where $d_{MV}$ is the metric given by 
%$$d_{MV}(X_t,X_{t'})=\min_{W}\max_{u} \EE[\langle %X_t-WX_{t'},u\rangle^2]^{1/2}=\min_{W} \left\|\EE[(X_t-WX_{t'})(X_t-WX_{t'})^T]\right\|^{1%/2},$$ the norm on the right being the spectral %norm. 
is nonbacktracking, 
so that $\varphi(t)=\varphi(t')$ implies $\varphi(s)=\varphi(t)$ for all $s\in[t,t']$, 
%which prevents loops from appearing in $\mathcal{M}=\varphi([0,T])$. 
and that $(\mathcal{M},d_{MV})$ is approximately Euclidean realizable. Recall that Theorem \ref{thm:manifold} states that under these conditions,
%, in the sense that there is a Lipschitz-continuous map $\psi:[0,T]\rightarrow\RR^c$ such that for some $c\geq0$, $$|d_{MV}(\varphi(t),\varphi(t'))-\|\psi(t)-\psi(t')\||\leq c|t-t'|\text{ for all }t,t'\in[0,T].$$
$\mathcal{M}$ is homeomorphic to an interval $[0,I]$. In particular, it is a topological 1-manifold with boundary. If $\varphi$ is injective, $\mathcal{M}$ remains a 1-manifold with boundary even when $\varphi$ is only $\alpha$-H\"{o}lder Euclidean realizable.
% Remarks: First, note that the regularity conditions are placed on $\psi$, which takes values in $\RR^c$, rather than on $\varphi$, which gives random variables as output. Since $\psi(t)$ is estimable from the data using CMDS, this means that we can visually inspect whether or not this assumption is plausible. Moreover, this gives us a clear interpretation of the points returned by CMDS: they are simply the Euclidean realization of the manifold $\mathcal{M}$ in the space of random variables, which means that we can simply view this as an approximately distance-preserving picture of those random variables, each of which captures the full state of the system with all of the given entities at any time $t$.

{\bf Proof of Theorem \ref{thm:manifold}}.
We consider the case where $\varphi$ is injective first, since this avoids some of the technical details of the more general case. We may first observe that $$d_{MV}(\varphi(t),\varphi(t'))\leq \|\psi(t)-\psi(t')\|+c|t-t'|^\alpha\leq (L+c)|t-t'|^{\alpha},$$ so $\varphi$ is also $\alpha$-H\"{o}lder continuous (or Lipschitz if $\alpha=1$). Since $\mathcal{M}$ is defined to be $\varphi([0,T])$, it is apparent that $\varphi$ is bijective. Now any closed subset $F\subseteq [0,T]$ is compact, so $\varphi(F)\subseteq \mathcal{M}$ is compact, hence closed in $\mathcal{M}$, and $\varphi$ is a closed map. In other words, $\varphi^{-1}$ is continuous, so $\varphi$ is itself the required homeomorphism.

In the case that $\varphi$ is nonbacktracking and $\alpha=1$, we define $L(t,\delta)=\sup\{d_{MV}(\varphi(x),\varphi(y))/|x-y|: x,y\in B(t,\delta), x\neq y\}$ for any $t\in(0,T)$ and $\delta>0$ such that $B(t,\delta)\subseteq[0,T]$. This is finite and bounded above by $L+c$ from the first part of the proof. We also define $L(t)=\inf \{L(t,\delta): \delta>0\}$, which is a finite, nonnegative number bounded above by $L+c$. We make the following observations, which are easily proved: (1) $L(t,\delta)$ is lower semicontinuous in $t$, for any $\delta>0$; (2) $L(t)$ is integrable.

Now we define $\gamma:[0,T]\rightarrow[0,I]$ via $\gamma(t)=\int_0^t L(s)\,\mathrm{d}s$, where $I=\int_0^T L(s)\,\mathrm{d}s.$ Since $\gamma$ is surjective, this allows us to define $\mu:[0,I]\rightarrow\mathcal{M}$ via $\mu(\gamma(t))=\varphi(t)$ for all $t\in[0,T]$. We now show that $\mu$ is well-defined and Lipschitz continuous. Let $t<t'\in[0,T]$, and given $\delta>0$, choose points $s_i$ such that $s_0=t, s_k=t'$, and $s_i<s_{i+1}<s_{i}+\delta$ for each $i$. Now we observe that
\begin{align*}
d_{MV}(\mu(\gamma(t)),\mu(\gamma(t')))&=d_{MV}(\varphi(t),\varphi(t'))\\
&\leq \sum_{i=0}^{k-1} d_{MV}(\varphi(s_i),\varphi(s_{i+1}))\\
&\leq \sum_{i=0}^{k-1} L(s_i,\delta)(s_{i+1}-s_i).
\end{align*}
Letting the partition $\{s_i\}$ of $[t,t']$ become arbitrarily fine, we see that this upper bound converges to the corresponding integral, giving $$d_{MV}(\mu(\gamma(t)),\mu(\gamma(t')))\leq \int_{t}^{t'} L(s,\delta)\,\mathrm{d}s.$$ Now taking an infimum over $\delta>0$ and applying dominated convergence, we see that $$d_{MV}(\mu(\gamma(t)),\mu(\gamma(t')))\leq \int_t^{t'} L(s)\,\mathrm{d}s=|\gamma(t')-\gamma(t)|.$$

Observe that $\mu$ is injective: if $\mu(x)=\mu(y)$ for some $x<y\in [0,I]$, then for $t<t'\in[0,T]$ with $\gamma(t)=x$ and $\gamma(t')=y$, we have that $$\varphi(t)=\mu(\gamma(t))=\mu(x)=\mu(y)=\mu(\gamma(t'))=\varphi(t'),$$ so $\varphi(s)=\varphi(t)$ for all $s\in[t,t']$. Then for $s\in (t,t')$, we can take $\delta>0$ small enough that $B(s,\delta)\subseteq(t,t')$, and since $d_{MV}(\varphi(a),\varphi(b))=0$ for all $a,b\in B(s,\delta)$, we see that $L(s,\delta)=0$, and thus $L(s)=0$, too. Now from the definition of $\gamma$, $$y-x=\gamma(t')-\gamma(t)=\int_t^{t'} L(s)\,\mathrm{d}s=0,$$ which contradicts the assumption that $x<y$.

Since $\mu:[0,I]\rightarrow\mathcal{M}$ is a continuous bijection, it is easy to see that $\mu$ is in fact a homeomorphism.\hfill$\square$\\

%\end{proof}
When the trajectories $X(\cdot,\omega):[0,T]\rightarrow\RR^d$ satisfy a  H\"{o}lder condition with square-integrable constant, then $\varphi$ also satisfies this continuity condition, as Theorem \ref{thm:smoothtraj} states.

\textbf{Proof of Theorem~\ref{thm:smoothtraj}:} 
To show that sufficiently smooth trajectories imply continuity for $\varphi$, note that
\begin{align*}
d_{MV}(X_t,X_s)&= \min_W \|\EE[(X_t-WX_s)(X_t-WX_s)^T]\|_2^{1/2}\\
&\leq \|\EE[(X_t-X_s)(X_t-X_s)^T]\|_2^{1/2}\\
&\leq \EE[\|X_t-X_s\|^2]^{1/2}\\
&\leq \EE[L(\omega)^2 |t-s|^{2\alpha}]^{1/2}\\
&= \|L\|_{L^2(\Omega)} |t-s|^{\alpha}.
\end{align*}
\hfill$\square$

Additional constraints on the probabilistic structure of the stochastic process can render the distance $d_{MV}$ simpler to compute. In Theorem \ref{thm:martingale}, we show that if $\varphi(t)=\gamma(t)+M_t$, where $\gamma:[0,T]\rightarrow\RR^d$ is Lipschitz continuous and $M_t$ is a martingale with certain variance constraints, then $\varphi$ is approximately $\alpha$-H\"{o}lder Euclidean realizable.
%satisfying $\|\mathrm{Cov}(M_t-M_s)\|\leq C(t-s).$ Under these conditions, $\varphi$ is approximately $\alpha$-H\"{o}lder Euclidean realizable with $\alpha=1/2$, and $\psi=\gamma$: that is, the Euclidean realization simply recovers the deterministic drift.

\textbf{Proof of Theorem~\ref{thm:martingale}:}
Suppose $\varphi(t)=\gamma(t)+M_t$, where $\varphi(t)=\gamma(t)+M_t$, and $\gamma:[0,T]\rightarrow\RR^d$ is Lipschitz continuous with $\gamma(t)=a(t) v$;  $M_t$ is a martingale satisfying $\|\mathrm{Cov}(M_t-M_s)\|_2\leq C(t-s).$
We expand $(X_t-WX_s)(X_t-WX_s)^\top$ using the decomposition
$$X_t-WX_s=M_t-M_s+\gamma(t)-W\gamma(s)+M_s-WM_s.$$ Since the increment $M_t-M_s$ is conditionally mean-zero given $\mathcal{F}_s$, and $M_s$ has mean zero, all cross terms vanish when we take the expected value. This leaves
\begin{align*}
\EE[(X_t-WX_s)(X_t-WX_s)^\top]&=\EE[(M_t-M_s)(M_t-M_s)^\top]+(\gamma(t)-W\gamma(s))(\gamma(t)-W\gamma(s))^\top\\
&+(I-W)\EE[M_sM_s^\top](I-W)^\top
\end{align*}
Plugging in $W=I$ and using the triangle inequality guarantees that
$$d_{MV}(X_t, X_s)^2 \leq \|\textrm{Cov}(M_t-M_s)\|_2 + \|\gamma(t)-\gamma(s)\|^2$$
as required.

Since $\gamma(s)\in\mathrm{span}(\gamma(t))$, we use the fact that the first and last terms are positive semidefinite to obtain the lower bound $(\|\gamma(t)\|-\|\gamma(s)\|)^2=\|\gamma(t)-\gamma(s)\|^2$, so $$\|\gamma(t)-\gamma(s)\|^2\leq d_{MV}(X_t,X_s)^2\leq \|\gamma(t)-\gamma(s)\|^2+C(t-s),$$ which completes the proof.
%\textcolor{red}{NEED TO ADD MINIMIZATION DETAILS}
\hfill$\square$

We now demonstrate the properties of the stochastic processes describe in Examples \ref{ex:integratedBM} and \ref{ex:BM}.

\textbf{Proof for Example ~\ref{ex:integratedBM}:}
Since $X_t=\gamma(t)+I_t$ with $I_t=\int_0^t B_t$, observe that 
\begin{align}
d_{MV}(X_t,X_{t'})^2&=\min_W \|\EE[(X_t-W X_{t'})(X_t-W X_{t'})^{\top}]\|_2\notag\\
&=\min_W \Bigg\|\EE\left[ \gamma(t)-W\gamma(t') + I_t-WI_{t'})(\gamma(t)-W\gamma(t')+ I_t-WI_{t'})^{\top}\right]\Bigg\|_2\notag\\
&=\min_W \Bigg\| (\gamma(t)-W\gamma(t'))(\gamma(t)-W\gamma(t'))^{\top}+\EE\left[(I_t-WI_{t'})(I_t-WI_{t'})^{\top}\right]\Bigg\|_2\label{eq:integratedBM}
\end{align}
We may expand $I_t-WI_{t'}$ as $I_t-I_{t'}+(I-W)I_{t'}$. Using the fact that $\EE[B_s|\mathcal{F}_{s'}]=B_{s'}$ whenever $s>s'$, we observe that 
\begin{align*}
\EE[I_{t'}(I_t-I_{t'})^\top]&=\EE[I_{t'}\EE[(I_t-I_{t'})^\top|\mathcal{F}_{t'}]]=\EE[I_{t'}(t-t')B_{t'}^\top]\\
&=(t-t')\int_{0}^{t'}\EE[B_sB_{t'}^{\top}]\,\mathrm{d}s=(t-t')\int_0^{t'} \EE[B_s B_s^\top]\,\mathrm{d}s\\
&=(t-t')\int_0^{t'} \sigma^2 s I\,\mathrm{d}s=\sigma^2[(t-t')t'^2/2] I.
\end{align*}
Similarly,
\begin{align*}
\EE[(I_{t}-I_{t'})(I_t-I_{t'})^\top]&=\int_{t'}^t\int_{t'}^t\EE[B_s B_{s'}^\top]\,\mathrm{d}s'\,\mathrm{d}s=\int_{t'}^t\int_{t'}^t \sigma^2 \min\{s,s'\} I\,\mathrm{d}s'\,\mathrm{d}s\\
&=\sigma^2 I \int_{t'}^t \int_{t'}^s s'\,\mathrm{d}s'+s(t-s)\,\mathrm{d}s=\sigma^2 I \int_{t'}^t (s^2-t'^2)/2+st-s^2\,\mathrm{d}s\\
&=\sigma^2 I \int_{t'}^t st-t'^2/2-s^2/2\,\mathrm{d}s\\
&=\sigma^2 I [(t^2-t'^2)t/2-(t-t')t'^2/2-(t^3-t'^3)/6]\\
&=\sigma^2 I [t^3/3-t'^2t+2t'^3/3]=\sigma^2[(t-t')^2(t+2t')/3] I
\end{align*}
Since $(t+2t')/3=(t-t')/3+t'$, we may write this as $\sigma^2[(t-t')^3/3+(t-t')^2t']I$, where the latter term equals $\int_{t'}^t\int_{t'}^t\EE[B_{t'}B_{t'}^\top]\,\mathrm{d}s'\,\mathrm{d}s$.

Therefore
\begin{align}
\EE&[(I_t-WI_{t'})(I_t-WI_{t'})^\top]\notag\\
&=\EE[(I_t-I_{t'})(I_t-I_{t'})^\top]+\EE[(I_t-I_{t'})I_{t'}^\top](I-W)^T+(I-W)\EE[I_{t'}(I_t-I_{t'})^\top]\notag\\
&\quad+(I-W)\EE[I_{t'}I_{t'}](I-W)^\top\notag\\
&=\sigma^2[(t-t')^2(t+2t')/3]I+\sigma^2[(t-t')t'^2/2](I-W)^T+\sigma^2[(t-t')t'^2/2](I-W)\notag\\
&\quad+\sigma^2[t'^3/3](I-W)(I-W)^\top\notag\\
&=\sigma^2[(t^3+t'^3)/3]I-\sigma^2[tt'^2/2-t'^3/6](W+W^\top)\notag\\
&=\sigma^2[(t-t')^2(t+2t')/3]I+\sigma^2[(t-t'/3)(t'^2/2)](2I-W-W^\top)\label{eq:ex1intmd}
\end{align}
The latter term is positive semidefinite, so the spectral norm of this matrix is minimized at $W=I$. For the term $\gamma(t)-W\gamma(t')=(at+b)v-(at'+b)Wv$, by Cauchy-Schwarz, we have
\begin{align*}
\|(\gamma(t)-W\gamma(t'))(\gamma(t)-W\gamma(t'))^\top\|_2&=\|\gamma(t)-W\gamma(t')\|^2\\
&=\|\gamma(t)\|^2+\|\gamma(t')\|^2-2\langle \alpha(t),W\gamma(t')\rangle\\
&\geq (\|\gamma(t)\|-\|\gamma(t')\|)^2=a^2(t-t')^2\|v\|^2,
\end{align*}
where the lower bound is achieved with $W=I$ since $\alpha(t)$ and $\alpha(t')$ are linearly dependent. In Equation~\ref{eq:integratedBM}, we obtain a lower bound by discarding the second term of Equation~\ref{eq:ex1intmd}. Since the remaining portion of Equation~\ref{eq:ex1intmd} is just a multiple of the identity, we use the identity (for $\beta>0$) $\|XX^\top+\beta I\|_2=\|XX^\top\|_2+\beta$ to finally obtain the equality
$$d_{MV}(X_t,X_{t'})^2= a^2(t-t')^2\|v\|^2+\sigma^2[(t-t')^2(t+2t')/3].$$ We see that $\varphi$ is approximately Euclidean realizable with $\psi(t)=\sqrt{a^2\|v\|^2+\sigma^2T}t$, since
$$\left|d_{MV}(X_t,X_{t'})^2-|\psi(t)-\psi(t')\|^2\right|=\sigma^2(T-(t+2t')/3)(t-t')^2\leq \sigma^2T(t-t')^2,$$
and thus using $a^2-b^2=(a-b)(a+b)$, we get
\begin{align*}
\left|d_{MV}(X_t,X_{t'})-\|\psi(t)-\psi(t')\|\right|&=\frac{\sigma^2(T-(t+2t')/3)(t-t')^2}{(\sqrt{a^2\|v\|^2+\sigma^2(t+2t')/3}+\sqrt{a^2\|v\|^2+\sigma^2T})|t-t'|}\\
&\leq C|t-t'|.
\end{align*}
\hfill$\square$

Note that $I_t$ is not a martingale, but as this example demonstrates, the stochastic term need not be. Moreover, the increased regularity of integrated Brownian motion guarantees approximate Lipschitz Euclidean realizability. If we consider processes expressible as the sum of a deterministic drift and standard Brownian motion, we retain $\alpha$-H\"{o}lder Euclidean realizability, as Example \ref{ex:BM} asserts.

\textbf{Proof for Example~\ref{ex:BM}:} 
%\textcolor{red}{ADD W MINIMIZATION DETAILS} 
We may proceed as in the proof of Theorem~\ref{thm:martingale}, obtaining $$X_t-WX_{t'}=\gamma(t)-W\gamma(t')+B_t-B_{t'}+(I-W)B_{t'}.$$ Since $B_t-B_{t'}$ and $B_{t'}$ are independent and have mean 0, we see that
\begin{align*}
d_{MV}(X_t,X_{t'})^2&=\min_W\|\EE[(X_t-WX_{t'})(X_t-WX_{t'})^\top]\|_2\\
&=\min_W \left\|(\gamma(t)-W\gamma(t'))(\gamma(t)-W\gamma(t'))^\top+\EE[(B_t-B_{t'})(B_t-B_{t'})^\top]\right.\\
&\qquad\quad\;\; \left.+(I-W)\EE[B_{t'}B_{t'}^\top](I-W)^\top\right\|_2\\
&=\min_W \left\|(\gamma(t)-W\gamma(t'))(\gamma(t)-W\gamma(t'))^\top+\sigma^2(t-t')I+\sigma^2t'(I-W)(I-W)^\top\right\|_2\\
&\leq (a(t)-a(t'))^2\|v\|^2+\sigma^2|t-t'|.
\end{align*}
Since the last term inside the norm is positive semidefinite, we obtain a lower bound given by 
$$d_{MV}(X_t,X_{t'})^2\geq \min_W \|(\gamma(t)-W\gamma(t'))(\gamma(t)-W\gamma(t'))^\top+\sigma^2(t-t')I\|_2.$$
Arguing as in Example~\ref{ex:integratedBM}, we see that this is minimized at $W=I$, proving that $$d_{MV}(X_t,X_{t'})^2=(a(t)-a(t'))^2\|v\|^2+\sigma^2|t-t'|.$$ Consider $\psi(t)=\|v\|a(t)$: then 
$$\left|d_{MV}(X_t,X_{t'})^2-\|\psi(t)-\psi(t')\|^2\right|=\sigma^2|t-t'|.$$ As before, this gives
\begin{align*}
\left|d_{MV}(X_t,X_{t'})-\|\psi(t)-\psi(t')\|\right|&=\frac{\sigma^2|t-t'|}{\sqrt{(a(t)-a(t'))^2\|v\|^2+\sigma^2|t-t'|}+|a(t)-a(t')|\|v\|}\\
&\leq C|t-t'|^{1/2},
\end{align*}
so $\varphi$ is approximately (1/2)-H\"{o}lder Euclidean realizable.\hfill$\square$

\subsection{Proofs and supporting results for Section~\ref{sec:estimation}}
\label{sec:proofs3}

We will make use of the following supporting lemmas in our proof of Theorem~\ref{thm:approximation}. The first says that the property of equicontinuity for functions is preserved under convex combinations; we omit the straightforward proof.

\begin{lemma}
\label{lem:convequicontinuity}
Let $(A,d)$ be a metric spaces, and let $\mathcal{C}$ be a collection of functions $f:A\rightarrow\RR$ such that for any $\epsilon>0$, there exists $\delta(\epsilon)>0$ such that for all $f\in\mathcal{C}$, $$\forall a,b\in A,\; d(a,b)<\delta(\epsilon)\Rightarrow |f(a)-f(b)|<\epsilon.$$ That is, $\mathcal{C}$ is an equicontinuous family of functions. Then the convex hull of $\mathcal{C}$,  $\mathrm{conv}(\mathcal{C})$, is an equicontinuous family of functions with this same $\delta(\epsilon)$. That is, given $f_1,\ldots,f_n\in \mathcal{C}$, $\lambda_1,\ldots,\lambda_n\geq0$ with $\sum_{i=1}^n\lambda_i=1$, $$\sum_{i=1}^n \lambda_i f_i$$ is uniformly continuous with the same modulus of continuity (at most) $\delta(\epsilon)$ for all $\epsilon>0$.
\end{lemma}

% \begin{proof}
% We argue by induction. When $n=1$, the statement is trivial, so we suppose $f,g\in\mathcal{F}$, $\lambda\in(0,1)$ (otherwise we just revert to the case $n=1$). When $d(a,b)<\delta(\epsilon)$, we know that
% $$|f(a)-f(b)|<\epsilon,\quad |g(a)-g(b)|<\epsilon,$$
% so 
% \begin{align*}
%     |\lambda f(a)+(1-\lambda)g(a) - (\lambda f(b)+(1-\lambda)g(b))|&\leq \lambda|f(a)-f(b)|+(1-\lambda)|g(a)-g(b)|\\
%     &\leq \lambda \epsilon+(1-\lambda)\epsilon=\epsilon.
% \end{align*}
% The inductive step now follows since if $\lambda_{n+1}\in(0,1)$, we have
% $$\sum_{i=1}^{n+1} \lambda_i f_i = (1-\lambda_{n+1})\left(\sum_{i=1}^n \frac{\lambda_i}{1-\lambda_{n+1}} f_i\right)+\lambda_{n+1}f_{n+1},$$ which reveals that $\bar{f}$ is just a convex combination of two functions that are already assumed equicontinuous with the same modulus of continuity $\delta(\epsilon)$ by the induction hypothesis, so we may apply the base case to reach the desired conclusion.
% \end{proof}

The following lemma states that when two functions are uniformly close, their maximum and minimum values must also be close.

\begin{lemma}
\label{lem:uniformmaxbound}
Let $f,g:A\rightarrow\RR$ be continuous functions on the compact metric space $(A,d)$, satisfying $$\max_{a\in A}\,|f(a)-g(a)|<\epsilon.$$ Then $$\left|\max_{a\in A}\,f(a)-\max_{b\in A}\,g(b)\right|<\epsilon,\quad \left|\min_{a\in A}\,f(a)-\min_{b\in A}\,g(b)\right|<\epsilon.$$
\end{lemma}

% \begin{proof}
% We see that if $a^*\in A$ is such that $f(a^*)=\max_{a\in A}f(a)$, then $g(a^*)>f(a^*)-\epsilon$, so $\max_{b\in A}g(b)\geq g(a^*)>\max_{a\in A}f(a)-\epsilon$. By symmetry, $f$ satisfies the same bound with $g$, which proves the result for the maxima. Taking $-f,-g$ and applying the first result proves the result for the minima.
% \end{proof}

Since the mean edge probability matrices of latent position graphs have low-rank structure, the underlying latent position matrices can be well estimated by the adjacency spectral embedding \citep{STFP-2011, athreya_survey}. In particular, the difference between the Procrustes-aligned adjacency spectral embeddings of two independent networks satisfies the following concentration (\cite{tang14:_semipar}; Theorem 3.3 and Corollary 3.4 in \cite{athreya2016numerical}):

\begin{lemma}
\label{lem:aseconcentration}
Let $F$ be an inner product distribution in $\RR^d$, and suppose $X_1,\ldots,X_n\sim F$ are an i.i.d. sample. Suppose that $\EE[X_1 X_1^\top]$ has rank $d$. If $A$ is the adjacency matrix of an RDPG with latent position matrix $\bX$ having rows $X_i^\top$, $1\leq i\leq n$, and $\hat{\bX}$ is its $d$-dimensional ASE, then there is a constant $C$ such that with overwhelming probability, $$\min_W \|\hat{\bX}-\bX W\|_2\leq C+O(log(n)/\sqrt{n}).$$
\end{lemma}

We now turn to our main result, Theorem~\ref{thm:approximation}, which says that with overwhelming probability, $d_{MV}(\hat{\bX}_t,\hat{\bX}_s)$ and $d_{MV}(X_t,X_s)$ are close. A crucial step is the proof of a concentration inequality for the scaled distance between the realized latent position matrices $\bX_t$ and $\bX_s$ and the maximum directional variation metric between $X_t$ and $X_s$, whose joint distribution is inherited from the latent position process $\varphi$. For a given $W\in\mathcal{O}^{d\times d}$ and $u\in\RR^d$ with $\|u\|=1$, classical results show that $\frac{1}{n}\|(\bX_t-\bX_s W)u\|^2$ concentrates around $\EE[\langle X_t-WX_s,u\rangle^2]$. The challenge, however, is the maximization over $u$ and minimization over $W$. This necessitates a \emph{uniform} concentration bound in $W,u$, which relies on a carefully-constructed cover for the compact set $\mathcal{O}^{d\times d}\times \mathcal{S}^d$. We show this pointwise concentration can be extended uniformly over small neighborhoods of a given $(W,u)$, after which a union bound gives the desired result.

\textbf{Proof of Theorem~\ref{thm:approximation}:}
Consider the matrices of true latent positions, $\bX_t$ and $\bX_s$, and denote the rows of these matrices as $(X^i,Y^i)$, $i=1,\ldots,n$. These rows these are an i.i.d. sample from some latent position distribution $F$. We first show that with overwhelming probability, we have the bound
\begin{equation}
\label{eq:truelatents}
\left|\hat{d}_{MV}(\bX_t,\bX_s)^2-d_{MV}(X^1,Y^11)^2\right|\leq \frac{\log(n)}{\sqrt{n}}.
\end{equation}
From the definition of $\hat{d}_{MV}$, we have
\begin{align*}
\min_{W\in\mathcal{O}^{d\times d}}\frac{1}{n}\|\bX_t-\bX_sW\|_2^2&=\min_W \max_{\|u\|=1}\frac{1}{n}\|(\bX_t-\bX_sW)u\|^2\\
&=\min_W \max_u \frac{1}{n}\sum_{i=1}^n \langle X^i-WY^i,u\rangle^2.
\end{align*}
Defining $f((x,y),W,u)=\langle x-Wy,u\rangle^2$, and $d((W,u),(W',u'))=\max\{\|W-W'\|_2,\|u-u'\|\}$ it is easy to show that $$|f((x,y),W,u)-f((x,y),W',u')|\leq 12d((W,u),(W',u')).$$ Moreover, $$Z_n(\omega,W,u)=\frac{1}{n}\sum_{i=1}^n \langle X^i(\omega)-WY^i(\omega),u\rangle^2\in \mathrm{conv}(\{f((x,y),\cdot,\cdot):x\in \mathrm{supp}(F_t), y\in\mathrm{supp}(F_s)\}).$$ If we define $\mu(W,u)=\EE[Z_n(\cdot,W,u)]$, then by Lemma~\ref{lem:convequicontinuity}, $\{Z_n(\omega,\cdot,\cdot)-\mu(\cdot,\cdot):n\geq 1,\omega\in \Omega\}$ is an equicontinuous family of functions in $(W,u)$ with modulus of continuity $\delta(\epsilon)\geq\epsilon/24$.

Let $(W,u)$ be fixed, and consider $\mathcal{N}=\mathcal{N}_{W,u}=\{(W',u'):d((W',u'),(W,u))<\delta(\epsilon/2)\}.$ If $|Z_n(\omega,W,u)-\mu(W,u)|<\epsilon/2$, then for $(W',u')\in\mathcal{N}$,
\begin{align*}
|Z_n(\omega,W',u')-\mu(W',u')|&\leq |Z_n(\omega,W',u')-\mu(W',u')-(Z_n(\omega,W,u)-\mu(W,u))|\\
&+|Z_n(\omega,W,u)-\mu(W,u)|\\
&<\epsilon/2+\epsilon/2=\epsilon.
\end{align*}
So we have the following containment:
\begin{align*}
\mathcal{A}_{n,W,u,\epsilon/2}&:=\{\omega: |Z_n(\omega,W,u)-\mu(W,u)|<\epsilon/2\}\\
&\subseteq \{\omega:|Z_n(\omega,W',u')-\mu(W',u')|<\epsilon\;\forall (W',u')\in \mathcal{N}\}\\
&=:\mathcal{A}_{n,\mathcal{N},\epsilon}
\end{align*}
Since $\langle X_1-WY_1,u\rangle^2\in[0,4]$, by Bernstein's inequality, we have for $\epsilon<1$ the bound
$$\PP[\mathcal{A}_{n,W,u,\epsilon/2}^c]\leq 2\exp\left(-\frac{(n\epsilon/2)^2/2}{n\mathrm{Var}(\langle X_1-WY_1,u\rangle^2)+(4/3)(n\epsilon/2)}\right)\leq 2\exp\left(-\frac{n\epsilon^2}{136}\right).$$ Then for given $\gamma>0$, this probability is $\leq \gamma$ for $n\geq (136/\epsilon^2)\log(2/\gamma)$, independent of the particular choice $(W,u)$. The previous containment says that for this same $n$, we also have 
\begin{equation}
\label{eq:probbound}
\PP[\mathcal{A}_{n,\mathcal{N},\epsilon}^c]\leq \gamma.
\end{equation}

By compactness of $\mathcal{O}^{d\times d}\times \mathcal{S}^d$, we may extract a finite subcover $\{\mathcal{N}_i\}_{i=1}^k$ from the open cover $\{\mathcal{N}_{W,u}\}_{(W,u)}$. In fact, we may take $k\leq (4d^{3/4}/\delta(\epsilon/2))^{2d}$, and since $\delta(\epsilon)\geq \epsilon/24$, this gives $k\leq (192d^{3/4}/\epsilon)^{2d}$.

Fix $\epsilon,\gamma>0$. If we let $N=\max\{N_i(\epsilon,\gamma/k): 1\leq i\leq k\}$, where $N_i(\epsilon,\gamma/k)$ ensures that the inequality (\ref{eq:probbound}) holds for the neighborhood $\mathcal{N}_i$ (with $\gamma/k$ on the right hand side), then in the worst case, we have $N=(136/\epsilon^2)\log(2k/\gamma)$, which from the bound on $k$ gives $N\leq (136/\epsilon^2)(\log(2/\gamma)+(2d)\log(192d^{3/4}/\epsilon)).$ But for $n\geq N$, we have 
$$\PP\left[\left(\bigcap_{i=1}^k \mathcal{A}_{n,\mathcal{N}_i,\epsilon}\right)^{c}\right]=\PP\left[\bigcup_{i=1}^k \mathcal{A}_{n,\mathcal{N}_i,\epsilon}^c\right]\leq \sum_{i=1}^k \PP[\mathcal{A}_{n,\mathcal{N}_i,\epsilon}^c]\leq k(\gamma/k)=\gamma.$$
% \begin{align*}
% \PP\left[\left(\bigcap_{i=1}^k \mathcal{A}_{n,\mathcal{N}_i,\epsilon}\right)^{c}\right]&=\PP\left[\bigcup_{i=1}^k \mathcal{A}_{n,\mathcal{N}_i,\epsilon}^c\right]\\
% &\leq \sum_{i=1}^k \PP[\mathcal{A}_{n,\mathcal{N}_i,\epsilon}^c]\\
% &\leq k(\gamma/k)=\gamma.
% \end{align*}
Thus 
\begin{equation}
\label{eq:intersectbound}
\PP\left[\bigcap_{i=1}^k \mathcal{A}_{n,\mathcal{N}_i,\epsilon}\right]=\PP[\{\omega: |Z_n(\omega,W,u)-\mu(W,u)|<\epsilon\;\forall(W,u)\}]\geq 1-\gamma.
\end{equation}

In particular, taking $\epsilon=\log(n)/\sqrt{n}$, we may take $\gamma=n^{-\zeta}$ for any $\zeta>0$ and find that for $n$ sufficiently large, the inequality (\ref{eq:intersectbound}) holds, since $N=o(n)$.

Now given $\omega$ for which $|Z_n(\omega,W,u)-\mu(W,u)|<\epsilon$ for all $(W,u)$, we see from Lemma~\ref{lem:uniformmaxbound} that $$\left|\max_u Z_n(\omega,W,u)-\max_{u'}\mu(W,u')\right|<\epsilon.$$ We may now regard $\max_u Z_n(\omega,W,u)$ as a function only of $W$, and similarly for $\mu$, so applying the lemma again yields $$\left| \min_W \max_u Z_n(\omega,W,u)-\min_{W'}\max_{u'}\mu(W',u')\right|<\epsilon.$$ From Inequality~\ref{eq:intersectbound}, we obtain the desired $$\PP\left[\left|\hat{d}_{MV}(\bX_t,\bX_s)^2-d_{MV}(X^1,Y^1)^2\right|\leq\frac{\log(n)}{\sqrt{n}}\right]\geq 1-n^{-c}.$$

Now by Lemma~\ref{lem:aseconcentration}, with overwhelming probability, the ASEs for the corresponding adjacency matrices satisfy for $r\in\{t,s\}$ $$\hat{d}_{MV}(\hat{\bX}_r,\bX_r)\leq C_r/\sqrt{n}+O(\log(n)/n).$$ This gives $$\left|\hat{d}_{MV}(\hat{\bX}_t,\hat{\bX}_s)-\hat{d}_{MV}(\bX_t,\bX_s)\right|\leq \frac{C_t+C_s}{\sqrt{n}}+O\left(\frac{\log(n)}{n}\right).$$ Since both terms are no larger than constant order (so certainly less than $\log(n)$), we can use $|a^2-b^2|=|a-b||a+b|$ to show that $$\left|\hat{d}_{MV}(\hat{\bX}_t,\hat{\bX}_s)^2-\hat{d}_{MV}(\bX_t,\bX_s)^2\right|=O\left(\frac{\log(n)}{\sqrt{n}}\right).$$ Combining this bound with Equation~\ref{eq:truelatents} completes the proof.

\textbf{Proof of Theorem~\ref{thm:cmds}:} Applying Theorem~\ref{thm:approximation} to each entry of the matrix $\widehat{\mathcal{D}}_{\varphi}$ and taking a union bound, we see that with overwhelming probability, $$\|\widehat{\mathcal{D}}_{\varphi}^{(2)}-\mathcal{D}_\varphi^{(2)}\|_F\leq \frac{m\log(n)}{\sqrt{n}}.$$ 
Recall the definitions of the CMDS matrices $\hat{E}_\varphi=-\frac{1}{2}P\hat{\mathcal{D}}_{\varphi}^{(2)}P$ and $E_\varphi=-\frac{1}{2}P\mathcal{D}_{\varphi}^{(2)}P.$ It is immediate that, with overwhelming probability, 
\begin{equation}\label{eq:E-hat-minus-E-frobnorm-bound}
\|\hat{E}_{\varphi}-E_{\varphi}|\|_F \leq \|\widehat{\mathcal{D}}_{\varphi}^{(2)}-\mathcal{D}_\varphi^{(2)}\|_F \leq \frac{m \log n}{\sqrt{n}}.
\end{equation}
Next, recall $-\frac{1}{2}P\mathcal{D}_{\psi}^{(2)}P=U S U^{\top}$. This is a rank $c$ matrix whose $c$th largest eigenvalue is $\lambda_c(E_\varphi)$ and for which all remaining eigenvalues are zero. By \cite{DK_usefulvariant}, we have 
\begin{align*}
\min_{R\in\mathcal{O}^{c\times c}}\|\hat{U}-UR\|_F&\leq 2^{3/2}\frac{\|P (\widehat{\mathcal{D}}_{\varphi}^{(2)}-\mathcal{D}_\psi^{(2)}) P^\top\|_F}{\lambda_c(E_{\varphi})}\\
&\leq \frac{2^{3/2}}{\lambda_c(E_{\varphi})}\left(\|\hat{E}_{\varphi}-E_{\varphi}\|_F+\|P(\mathcal{D}_{\varphi}^{(2)}-\mathcal{D}_{\psi}^{(2)}) P^\top\|_F\right)\\
&\leq \frac{2^{3/2}}{\lambda_c(E_{\varphi})}\left(\frac{m \log(n)}{\sqrt{n}}+\left(\sum_{i=c+1}^{m}\lambda_i^2(E_{\varphi})\right)^{1/2}\right).
\end{align*}
Recall we denote this upper bound by $B=B(n,m,c)$, where
$$B=B(n,m,c)=\frac{2^{3/2}}{\lambda_c(E_{\varphi})}\left(\frac{m \log(n)}{\sqrt{n}}+\left(\sum_{i=c+1}^{m}\lambda_i^2(E_{\varphi})\right)^{1/2}\right).$$

%For the bound for scaled distances, %follows as in \cite{lyzinski15_HSBM}, 
%note that $\|U^\top (\hat{E}_{\varphi}-E_{\varphi})U\|_F\leq \| \hat{E}_{\varphi}-E_{\varphi}\|_F$. 
Now we prove the bound for the scaled eigenvectors. Set $\nu=\frac{m\log(n)}{\sqrt{n}}$ for the upper bound on the Frobenius norm of the noise term $\|\hat{E}_{\varphi}-E_{\varphi}\|_F$. Following \cite{lyzinski15_HSBM}, put $W_1 \Sigma W_2^{\top}$ as the singular value decomposition of $U^T\hat{U}$ and $R=W_1W_2^T$. We see that
\begin{align*}
\|U^T\hat{U}-R\|_F&\leq \|\Sigma-I\|_F = \left(\sum_{i=1}^c (1-\cos(\theta_i))^2\right)^{1/2}\leq \sum_{i=1}^c(1-\cos(\theta_i))\\
&\leq \sum_{i=1}^c (1-\cos^2(\theta_i))= \sum_{i=1}^c \sin^2(\theta_i)= \|\sin \Theta(U,\hat{U})\|_F^2 \leq 2B^2.
\end{align*}
where the final inequality on the $\sin(\Theta)$ distance again follows from \cite{DK_usefulvariant}. Consider $R\hat{S}-S R$. We may expand this as
$$
R\hat{S}-S R = (R-U^T\hat{U})\hat{S}+U^T(\hat{E}_\varphi-E_\varphi)\hat{U}+S(U^T\hat{U}-R).
$$
By the triangle inequality, $\|\hat{S}\|\leq \|S\|+\|\hat{E}_{\varphi}-E_{\varphi}\|,$ so we have
\begin{align}\label{eq:intertwining_bound}
\|R\hat{S}-S R\|_F &\leq \|R-U^T\hat{U}\|_F(\|\hat{S}\|+\|S\|)+\|U^T(\hat{E}_\varphi-E_\varphi)\hat{U}\|_F\notag\\ 
&\leq \|R-U^T\hat{U}\|_F(2\|S\|+\|\hat{E}_\varphi-E_\varphi\|)+\|U\|\|\hat{E}_\varphi-E_\varphi\|_F\|\hat{U}\| \notag \\ 
& \leq \|R-U^T\hat{U}\|_F(2\|S\|+\|\hat{E}_\varphi-E_\varphi\|_F)+\|\hat{E}_\varphi-E_\varphi\|_F \notag \\
&\leq 4B^2\lambda_1(E_\varphi)+(1+2B^2)\nu.
\end{align}

To bound $R\hat{S}^{1/2}-S^{1/2}R$, we observe that the $i,j$ entry of this matrix is just
$$ R_{ij}\frac{\lambda_j(\hat{S})-\lambda_i(S)}{\lambda_j^{1/2}(\hat{S})+\lambda_i^{1/2}(S)}.$$
By Weyl's inequality, $\lambda_j(\hat{S})\geq \lambda_j(S)-\|\hat{E}_{\varphi}-E_{\varphi}\|$, and thus on the high-probability set we are considering, we deduce that  $\lambda_j(\hat{S}) \geq \lambda_c(E_{\varphi})-\nu$. Since $\lambda_c(E_{\varphi})$ is some positive constant, the bound in Eq.~\ref{eq:E-hat-minus-E-frobnorm-bound} guarantees that with high probability for $n$ sufficiently large, $\lambda_j(\hat{S})>0$. Thus
$$ \|R\hat{S}^{1/2}-S^{1/2}R\|_F \leq \frac{\|R\hat{S}-SR\|_F}{\lambda_c^{1/2}(E_\varphi)} \leq \frac{4B^2\lambda_1(E_{\varphi})+(1+2B^2)\nu}{\lambda_c^{1/2}(E_\varphi)}.$$

Hence, with high probability, 
\begin{align*}
\|\hat{U}\hat{S}^{1/2}-US^{1/2}R\|_F&\leq \|(\hat{U}-UR)\hat{S}^{1/2}\|_F+\|U(R\hat{S}^{1/2}-S^{1/2}R)\|_F\\
&\leq B\left(\lambda_1^{1/2}(E_\varphi)+\frac{\nu}{2\lambda_1^{1/2}(E_\varphi)}\right)+ \frac{4B^2\lambda_1(E_\varphi)+(1+2B^2)\nu}{\lambda_c^{1/2}(E_\varphi)}\\
&\leq B\lambda_1^{1/2}(E_{\varphi})\left(2+4B\kappa^{1/2}+(1+2B)\frac{\nu}{\lambda_c(E_{\varphi})}\right),
\end{align*}
where, to bound $\|\hat{S}\|^{1/2}$, we again use Weyl's inequality and the fact that for positive $x$ and $y$, $\sqrt{x}\leq \sqrt{y}+\frac{|x-y|}{2\sqrt{y}}$. Here, $\kappa=\lambda_1(E_\varphi)/\lambda_c(E_\varphi)$ denotes the condition number of the low-rank projection of $E_\varphi$.
\hfill$\square$

\textbf{Proof of Corollary~\ref{cor:uniform}:}
We bound the difference between the distance matrices using the Lipschitz realizability assumption:
\begin{align*}
\|\mathcal{D}_{\varphi}^{(2)}-\mathcal{D}_{\psi}^{(2)}\|_F^2&\leq \sum_{i,j=1}^{m} (d_{MV}(\varphi(t_i),\varphi(t_j))^2-\|\psi(t_i)-\psi(t_j)\|_2^2)^2\\
&\leq (2A)^2\sum_{i,j=1}^{m} (d_{MV}(\varphi(t_i),\varphi(t_j))-\|\psi(t_i)-\psi(t_j)\|_2)^2\\
&\leq (2AB)^2\sum_{i,j=1}^{m} |t_i-t_j|^2\\
&\leq (2AB)^2(T/m)^2\sum_{i,j=1}^{m}|i-j|^2. 
\end{align*}
Some algebra yields that the latter sum equals $m^2(m^2-1)/6$, which yields the given bound.
\hfill$\square$

\subsection{Mirror estimates for deterministic-drift-plus-noise latent position processes}\label{sec:drift-plus-noise_LPP}
In this section, we provide estimation results for a time series of networks with latent position process given by Example~\ref{ex:BM}, in which the latent positions process follows
$X_t=\gamma(t)+B_t$, where $B_t$ is a $2$-dimensional Brownian motion, and $\gamma:[0,T] \rightarrow \mathbb{R}^d$ is a Lipschitz continuous function of the form $\gamma(t)=a(t)v$. We consider the case when $a(t)$ is a linear curve: $a(t)=c_1t +c_2$, with $c_1, c_2$ constants, and the case in which $a(t)$ is a quadratic curve: $a(t)=c_1t^2+c_2$. In both, we take $v=(1/\sqrt{2}, 1/\sqrt{2})$ and we choose constants, a scaling of Brownian motion and a time interval for which the result values of $X_t$ vectors remain in the first quadrant with overwhelming probability. In particular, for the linear drift, we consider $c_1=1/50$ and $c_2=1/10$; for the quadratic drift, we considered $c_1=1/1000, c_2=1/10$. The Brownian motion was scaled by a factor of approximately $1/1000$. We generated 30 networks, each on $n=2000$ nodes.

In Figure~\ref{fig:example1_simulation_scree}, we see that in both the linear and quadratic case, it is reasonable to consider classical multidimensional scaling of the estimated distance matrix $\hat{\mathcal{D}}_{\varphi}$ into one dimension. 
\begin{figure}
    \includegraphics[width=0.5\textwidth]{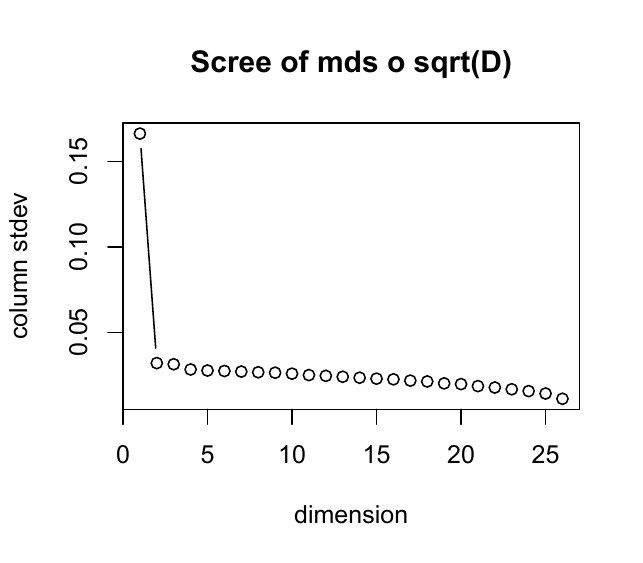}
    \includegraphics[width=0.5\textwidth]{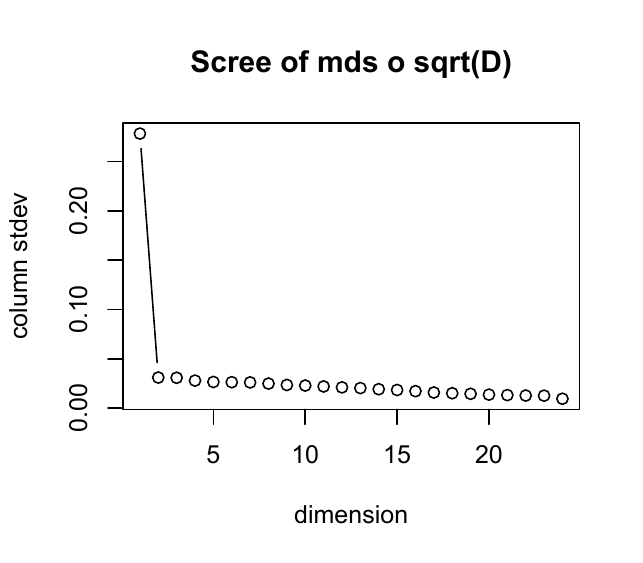}
    \caption{Left panel shows scree plot for estimated distance matrix $\hat{\mathcal{D}}_{\varphi}$ for the case of a linear drift. Right panel shows scree plot for estimated distance matrix $\hat{\mathcal{D}}_{\varphi}$ for quadratic drift.}
    \label{fig:example1_simulation_scree}
\end{figure}
As Figure \ref{fig:example1_simulation_estimated_mirror} shows, if we plot that MDS dimension over time, we observe a curve quite close to the actual mirror $\psi=\|v\|a(t)$ in both the linear and quadratic case.
\begin{figure}
    \includegraphics[width=0.5\textwidth]{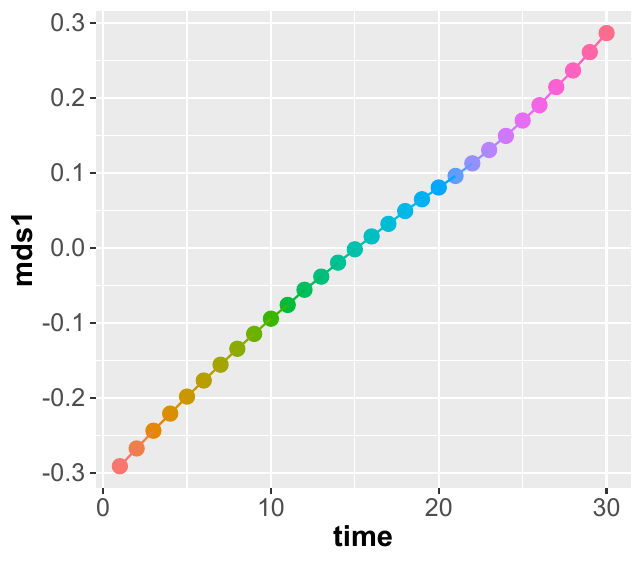}
    \includegraphics[width=0.5\textwidth]{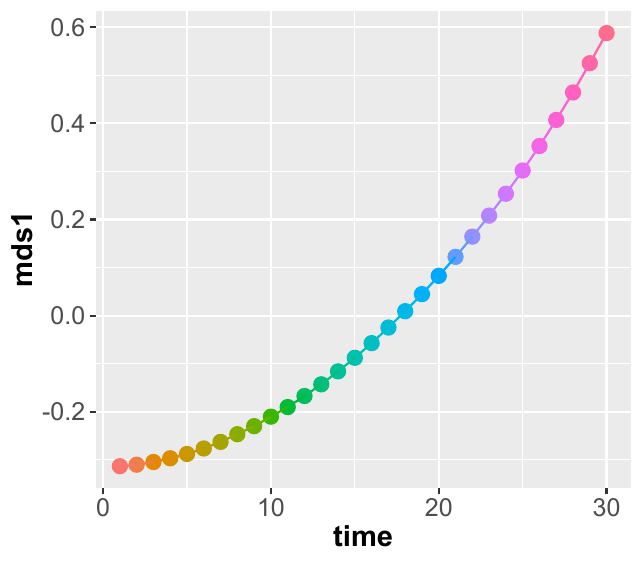}
    \caption{Left panel shows first dimension of the CMDS for estimated distance matrix $\hat{\mathcal{D}}_{\varphi}$ for the case of a linear drift. Right panel shows first dimension of CMDS for estimated distance matrix $\hat{\mathcal{D}}_{\varphi}$ for quadratic drift.}
    \label{fig:example1_simulation_estimated_mirror}
\end{figure}

\subsection{Mirror estimation for evolving stochastic blockmodels with varying connectivity}
\label{sec:sbms}
To illustrate the estimation of a mirror and its localization properties in a concrete case, we consider a time-series of two-community stochastic blockmodels with varying block connectivity matrix $B_t$. Let 
$$
B^1 = \begin{bmatrix}1/2&1/3\\ 1/3&1/2\end{bmatrix},\quad B^2=\begin{bmatrix}1/2&1/2\\1/2&1/2\end{bmatrix},\quad B^3=\begin{bmatrix}1/2&1/3\\1/3&1/3\end{bmatrix}.
$$
Then the block connectivity matrix is defined as
$$
B_t=\begin{cases}
(1-t) B^1+tB^2&t\in[0,1]\\
(2-t) B^2+(t-1)B^3&t\in[1,2]\\
(3-t) B^3 +(t-2)B^1&t\in[2,3].
\end{cases}
$$
We note that for $t=1,$ the block connectivity matrix only has rank 1, whereas for all other times, this matrix has rank 2. This has important consequences, which we discuss further in what follows.

To generate our network time series, we take thirty equally spaced times $t$ in the interval from $0$ to $3$, and for each $t$, we simulate a stochastic blockmodel network $G_t$ on $n=2000$ nodes with block probability matrix $B_t$, where $n/2=1000$ vertices belong to Cluster 1 and the other $1000$ vertices belong to Cluster 2.  As $t \in (0,1)$, we see a steady shift from in the block probability matrix from $B^1$ to $B^2$; and similarly for $t \in (1,2)$ and $t \in (2,3)$. 

Since the latent positions are known in this simulation, we can compute both the true the $d_{MV}$ distance and its realization-based estimate. Doing so, we get the matrices $\mathcal{D}_{\varphi}$ and $\mathcal{\hat{D}}_{\psi}$. The left panel of Figure~\ref{fig:sbms-distances} shows that the two matrices coincide fairly well outside of the change at $t=1$, when the the rank $2$ stochastic block model collapses into a rank $1$ Erd\"os-Renyi network, which constitutes a model misspecification: all networks are not, in fact, realizations of a constant rank $d$ random dot product graph. A scree plot of both the true and estimated dissimilarities suggests classical multidimensional scaling into $c=2$ dimensions provides a reasonable Euclidean approximation of both dissimilarities. Plotting the first and second dimensions of this embedding into two dimensions, we get the plots in Figure~\ref{fig:sbms-mds1and2}. It is striking that the the first dimension of the scaling is well-estimated, and the second dramatically less so. 

\begin{figure}[ht!]
    \includegraphics[width=0.5\textwidth]{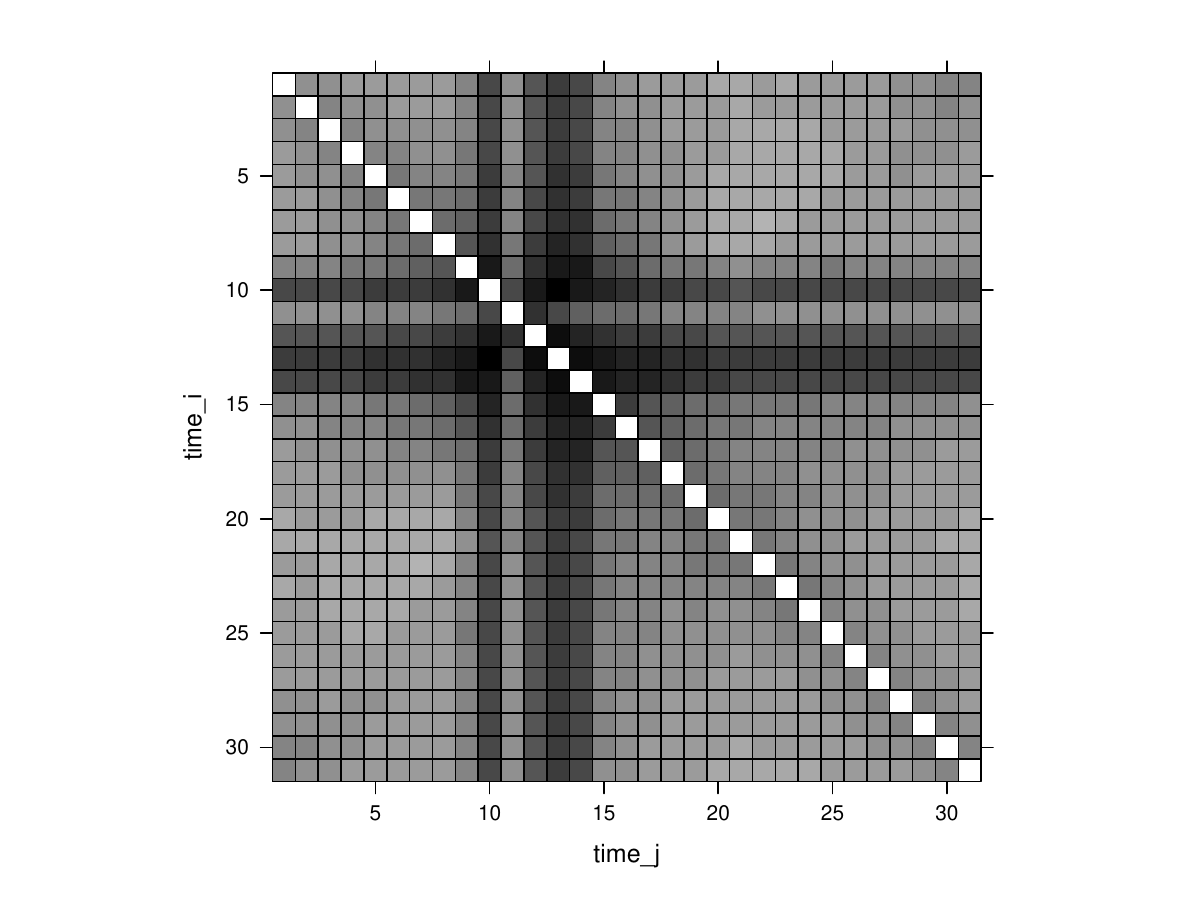}
    \includegraphics[width=0.5\textwidth]{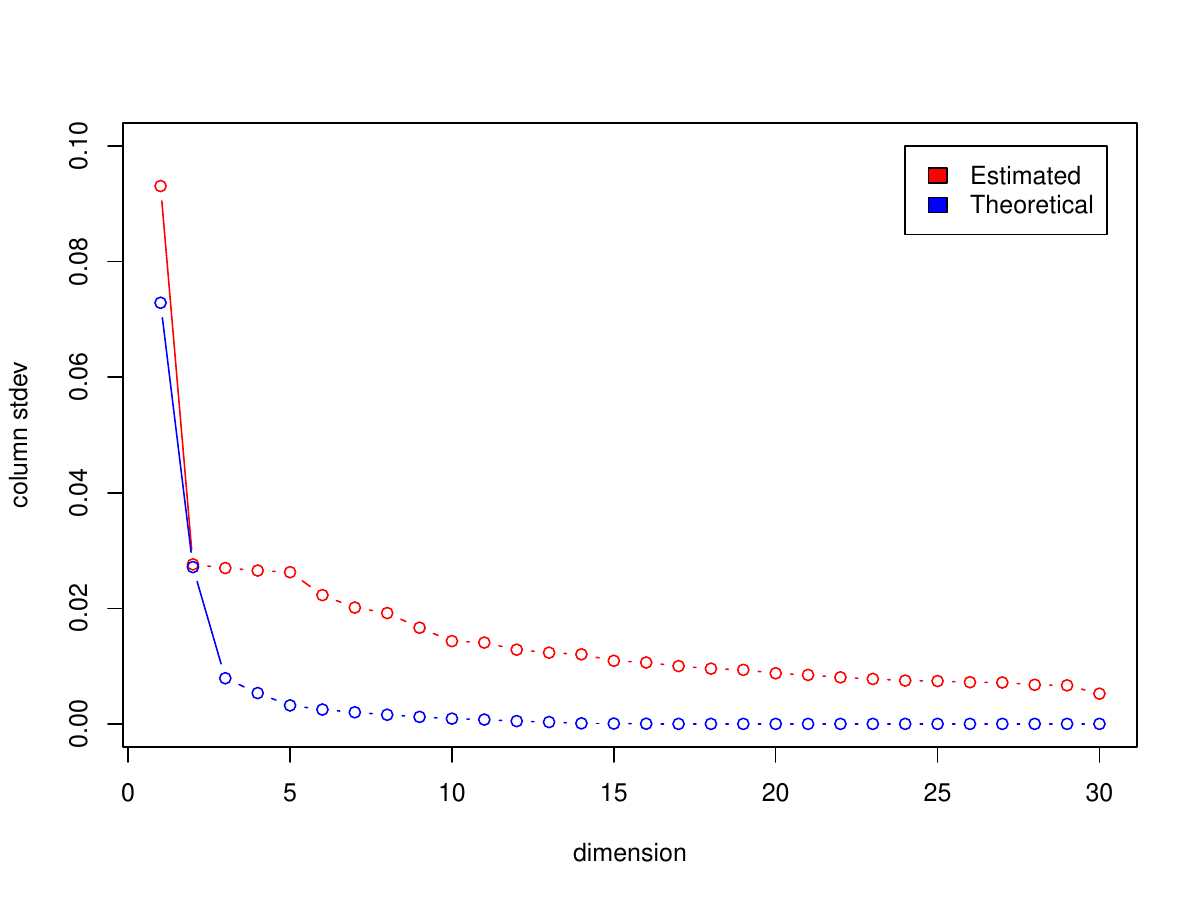}
    \caption{Left panel shows heatmap comparison of theoretical and estimated distance matrices $\mathcal{D}_{\varphi}$ and $\hat{\mathcal{D}}_{\varphi}$. Right panel shows scree plots for theoretical and estimated distance matrices $\mathcal{D}_{\varphi}$ and $\hat{\mathcal{D}}_{\varphi}$}
    \label{fig:sbms-distances}
\end{figure}

\begin{figure}[ht!]
    \includegraphics[width=0.5\textwidth]{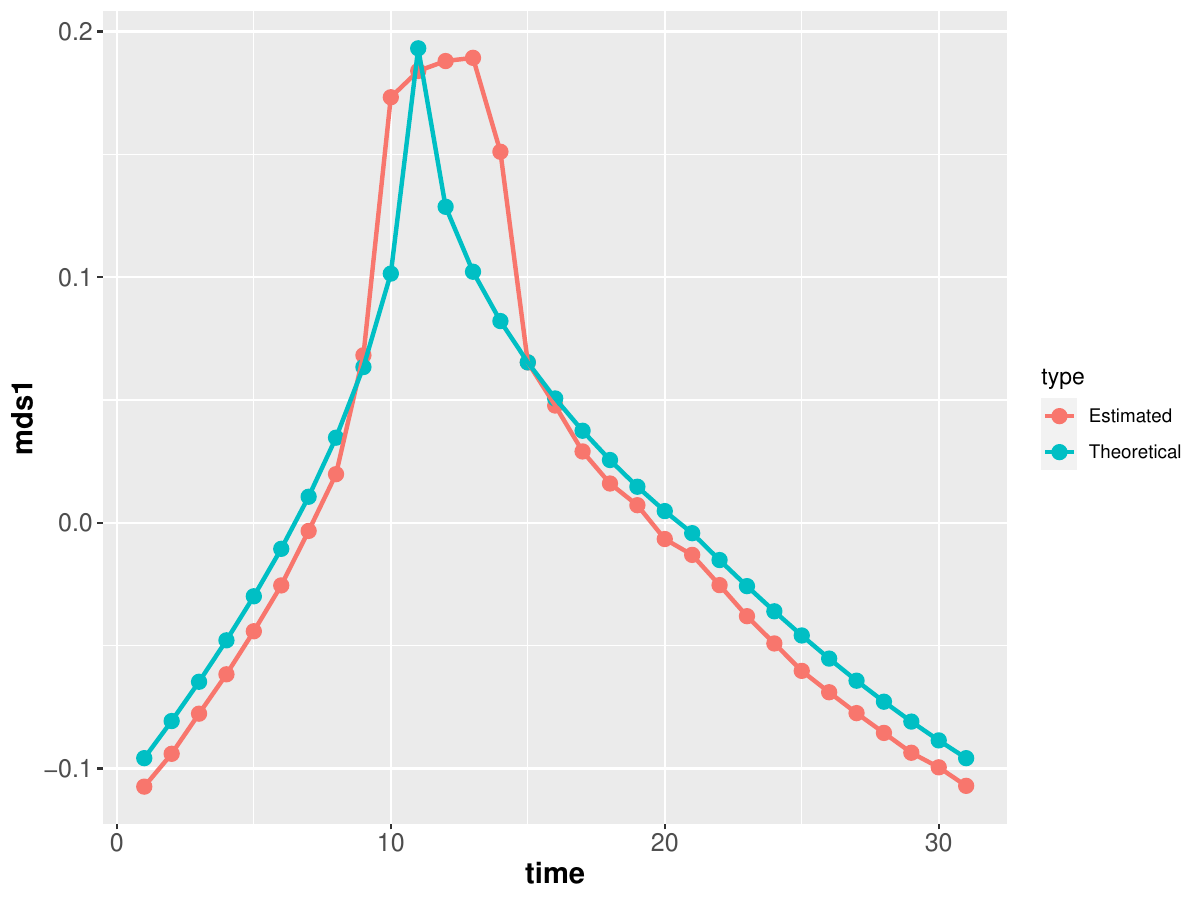}
    \includegraphics[width=0.5\textwidth]{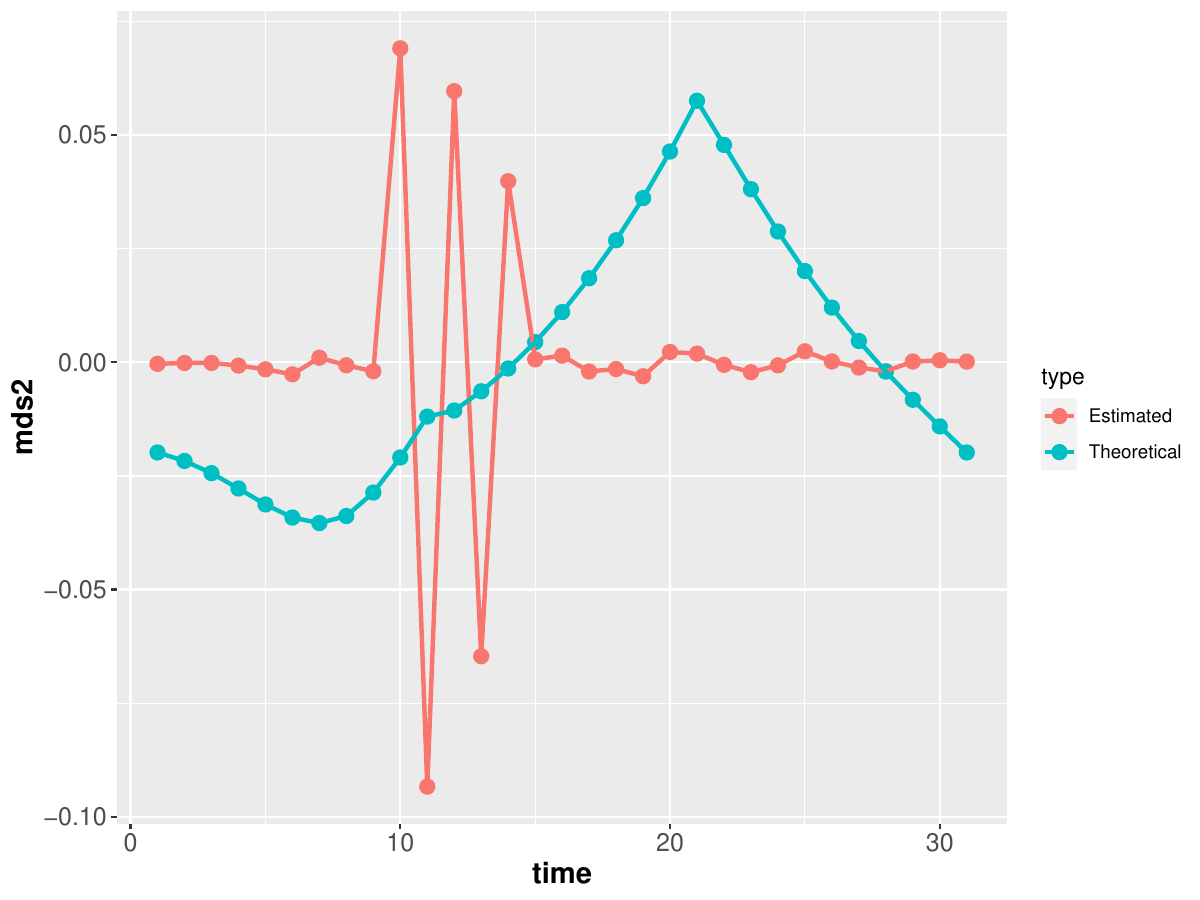}
    \caption{Comparison of first MDS dimension (left) and second MDS dimension (right) against time for theoretical and estimated pairwise distance matrices $\mathcal{D}_{\varphi}$ and $\hat{\mathcal{D}}_{\varphi}$. Note that despite the model misspecification for the graph at time $t=11$, the first component of the mirror is relatively accurate across the whole time interval. On the other hand, the second mirror component reflects the model misspecification around this time, rather than reflecting the theoretical distances.}
    \label{fig:sbms-mds1and2}
\end{figure}

The change in underlying rank for the $B_t$ matrices at $t=1$ constitutes an illuminating misspecification.  Such a stark shift in the rank corresponds to a type of underlying network change a mirror should detect, even if the hypotheses for our consistency results may not be satisfied.  Indeed, the true mirror does detect this with a cusp in its first embedding dimension, one that is replicated (approximately) by the corresponding plot for the top MDS dimension of $\hat{\mathcal{D}}_{\varphi}$.

The second embedding dimension for the case of the estimated distances (the red curve in the right panel of Figure \ref{fig:sbms-mds1and2}) reflects the noise in the second dimension of the adjacency spectral embedding for an Erd\"os-Renyi network.  Because an ER graph is a one-dimensional RDPG, the second dimension of the adjacency spectral embedding is driven by noise. This noise corrupts the accuracy of the estimated distance measure and leads to marked and distinct oscillations in the second MDS dimensions. These oscillations are not present on time intervals far removed from this changepoint.

\subsection{Additional visualizations and network statistics for real data communication networks }
\label{sec:realdataviz}

In Figure \ref{fig:realdataviz}, we provide additional visualizations of the organizational communication networks for January, May, and September of 2019 and 2020, allowing for a more detailed view of the evolution of the subcommunities over these two years. In contrast to multiple network visualizations over time, the mirror approach gives a much lower-dimensional and more quantitative signature of the changes in the networks. As such, while these images may be instructive for exploratory data analysis, they are much less useful for localization of changepoints compared to our mirror approach.

\begin{figure}
\centering
\begin{subfigure}{0.31\textwidth}
\includegraphics[width=\textwidth]{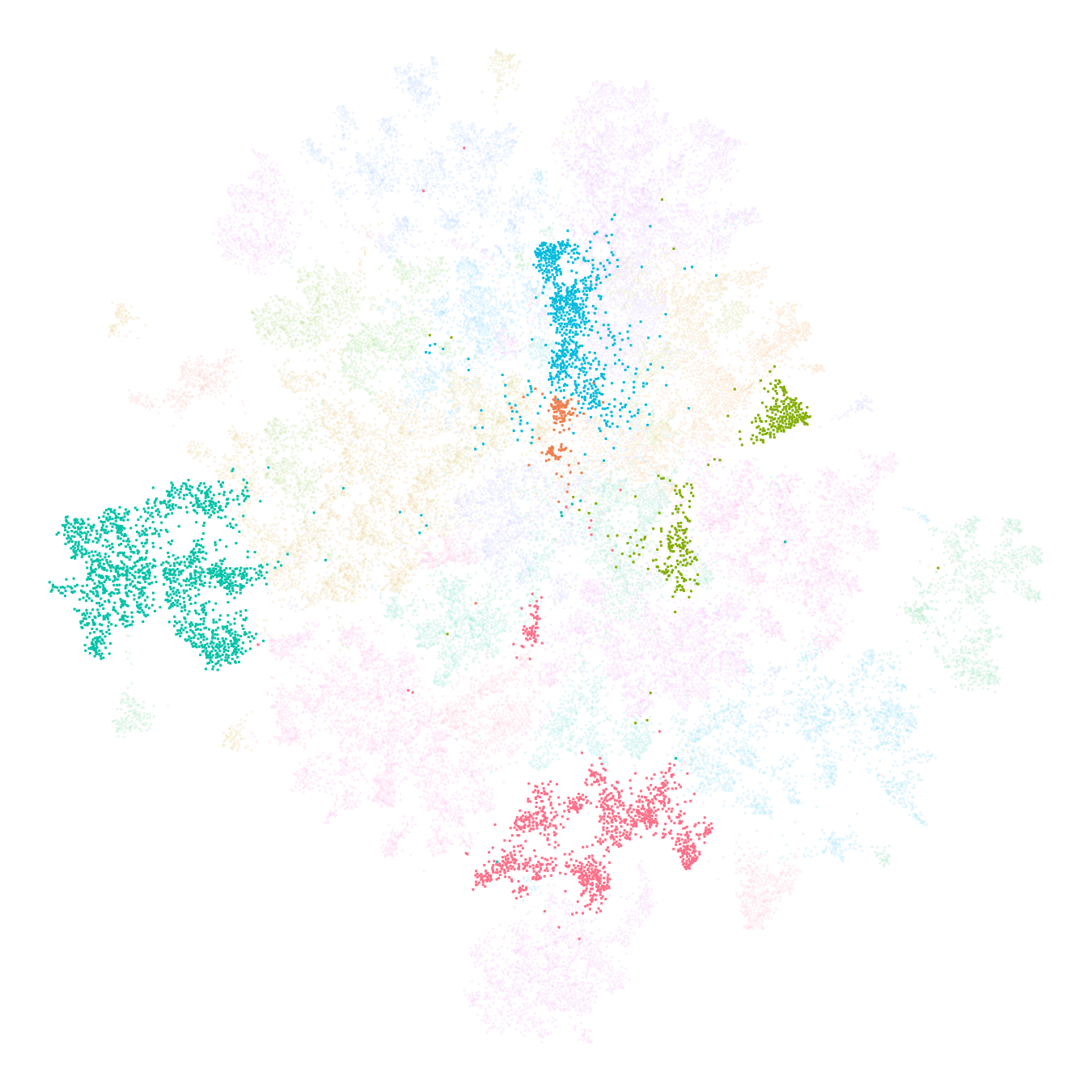}
\caption{January 2019}
\end{subfigure}
\begin{subfigure}{0.31\textwidth}
\includegraphics[width=\textwidth]{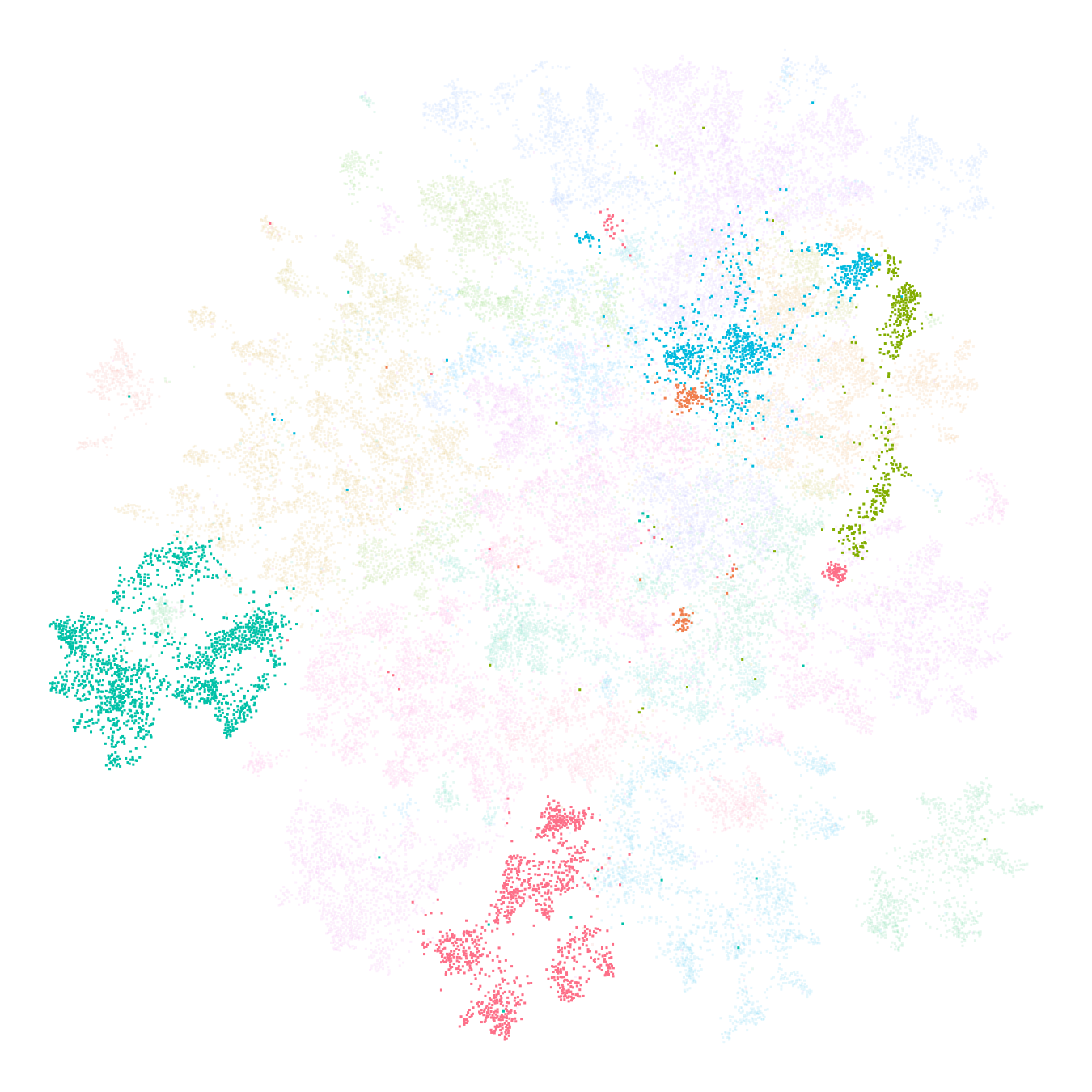}
\caption{May 2019}
\end{subfigure}
\begin{subfigure}{0.31\textwidth}
\includegraphics[width=\textwidth]{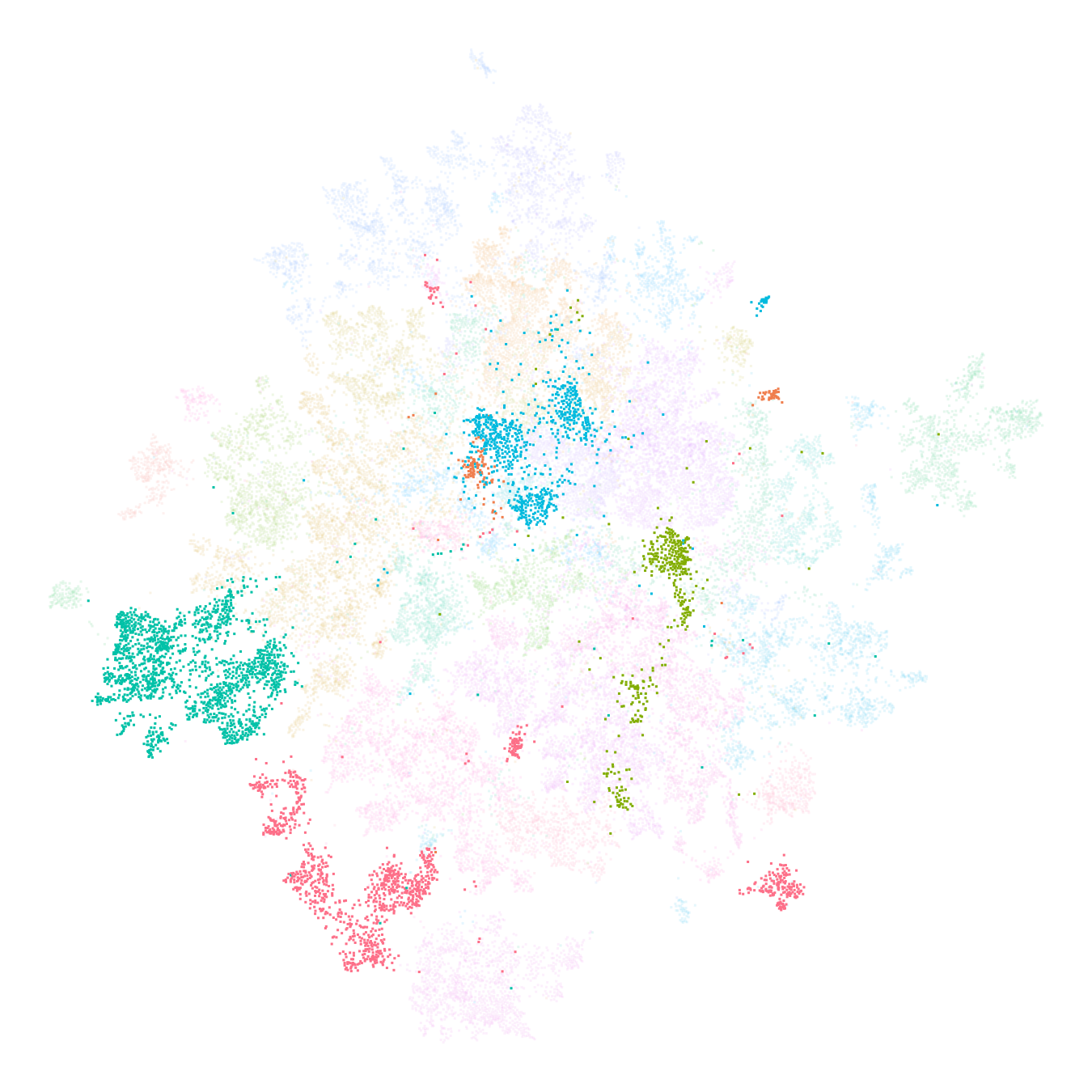}
\caption{September 2019}
\end{subfigure}

\begin{subfigure}{0.31\textwidth}
\includegraphics[width=\textwidth]{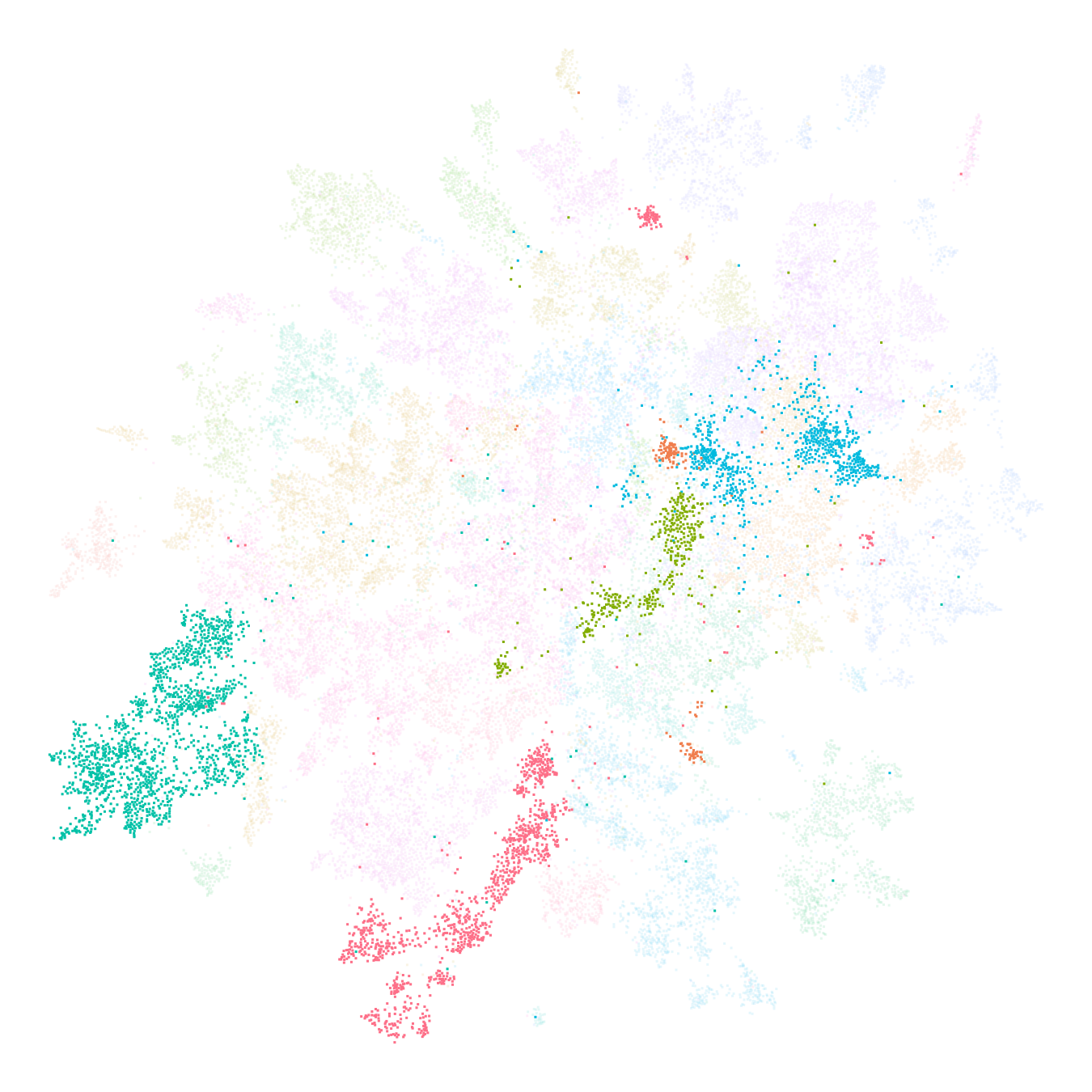}
\caption{January 2020}
\end{subfigure}
\begin{subfigure}{0.31\textwidth}
\includegraphics[width=\textwidth]{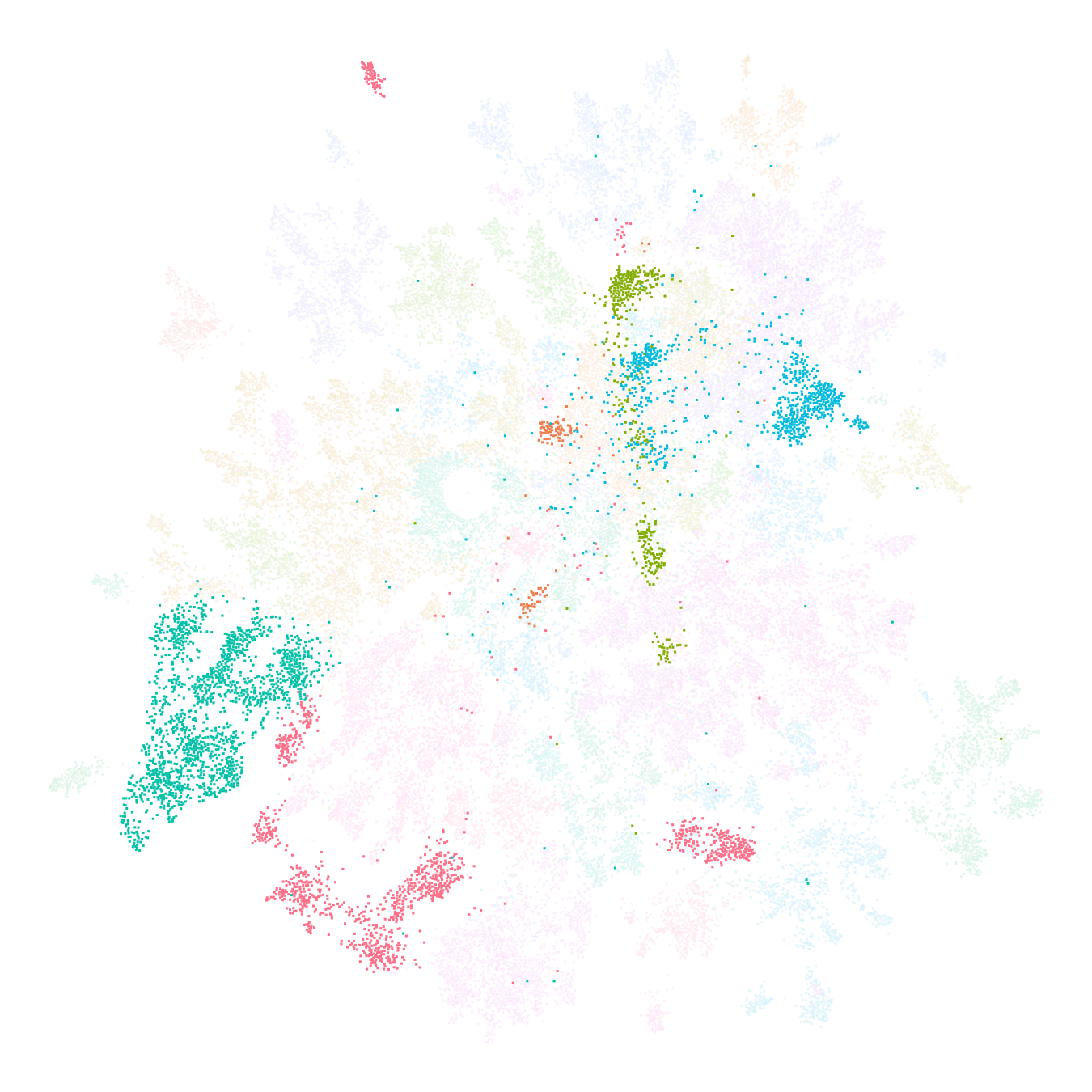}
\caption{May 2020}
\end{subfigure}
\begin{subfigure}{0.31\textwidth}
\includegraphics[width=\textwidth]{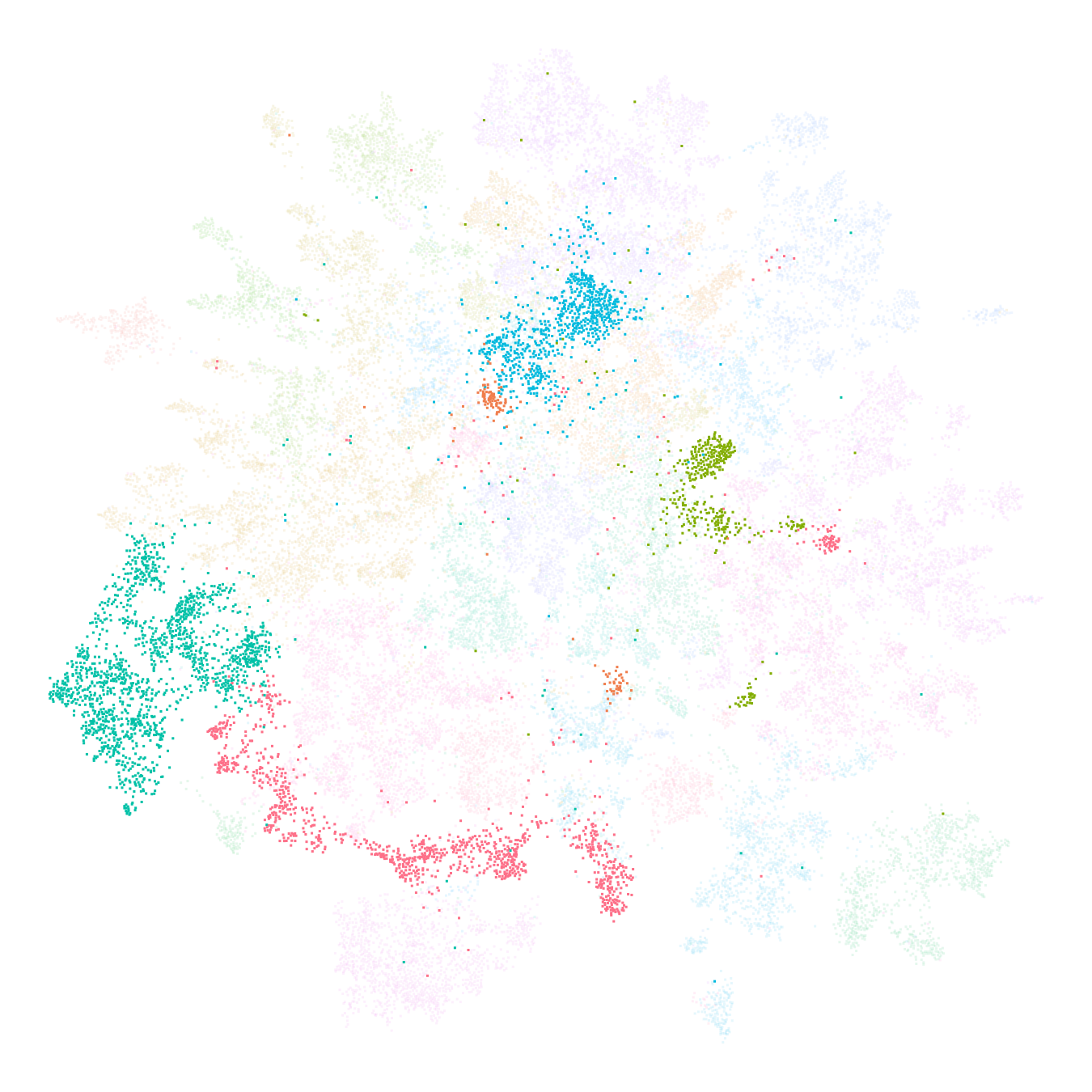}
\caption{September 2020}
\end{subfigure}
\caption{Additional visualizations of our real data organizational networks over the two year period, with 5 communities of various sizes highlighted, including the two communities from Figure~\ref{fig:org_sci_example}. We see that the different communities evolve in various ways over this time period, with some communities (like green) changing little over time, while others exhibit much larger changes, complementing the findings from Figure~\ref{fig:subcommunities}.}
\label{fig:realdataviz}
\end{figure}

 In Figure~\ref{fig:graphstats}, we plot a collection of other summary statistics, namely edge counts, maximum degree, median degree, and modularity, for each network over time. Since such statistics consider each network separately, these summary statistics exhibit greater variance than the ISOMAP embedding of the mirror (Figure~\ref{fig:org_sci_us} right panel, or Figure~\ref{fig:subcommunities}, bottom right panel). 
In addition, seasonal effects play a greater role in these plots, which add to the difficulty in detecting the changepoints. Note that in contrast to the mirror visualizations in Figure \ref{fig:subcommunities}, none of the plots in Fig. \ref{fig:graphstats} allows for easy qualitative visualization of two important changepoints driven by company policy at the start of the pandemic restrictions (Spring 2020) and the change in the imposition of restrictions from short-term to open-ended and longer-term (July 2020).

\begin{figure}[h]
    \centering
    \includegraphics[width=0.7\textwidth]{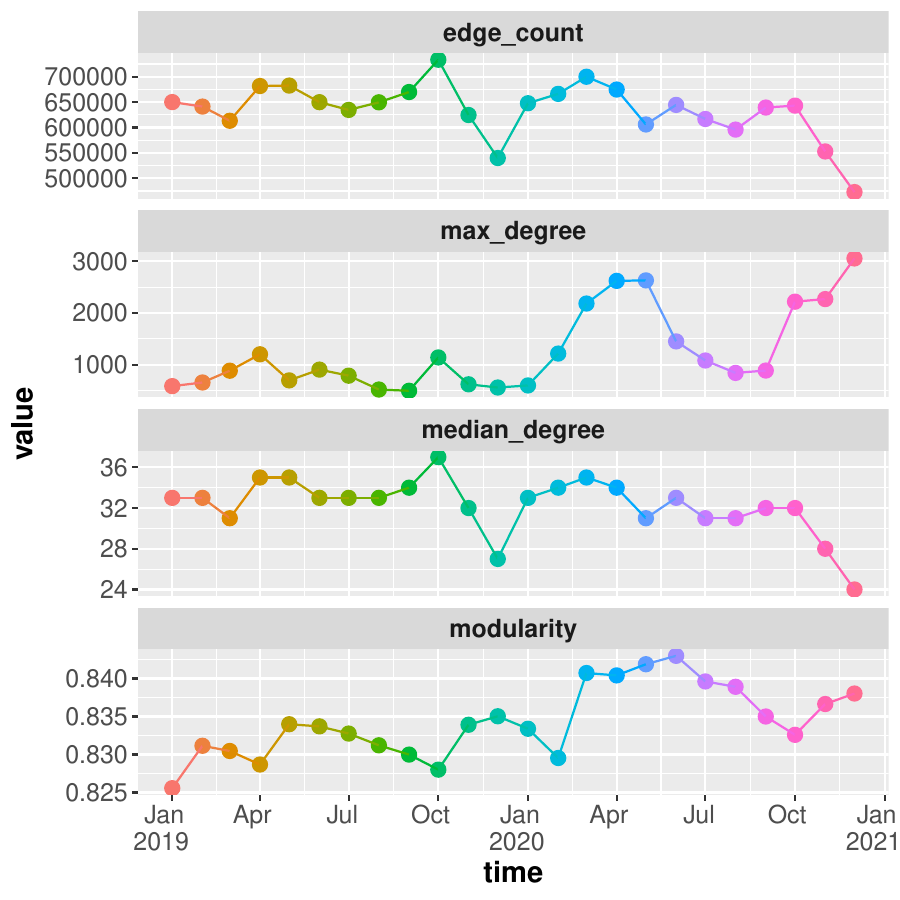}
    \caption{Additional graph summary statistics over time. These methods consider each graph separately, rather than our mirror approach which accounts for dependence across time: as a result, the curves show greater variance and seasonal effects, obfuscating the changepoints that are captured by the mirror in Figure \ref{fig:subcommunities}}
    \label{fig:graphstats}
\end{figure}

\end{document}